\documentclass[journal]{IEEEtran}
\IEEEoverridecommandlockouts
\usepackage{makecell}
\usepackage{cite}
\usepackage{amsmath,amssymb,amsfonts}

\usepackage{algorithmic}
\usepackage{graphicx}
\usepackage{textcomp}
\newtheorem{theorem}{Theorem}

\newtheorem{remark}{Remark}
\usepackage{bbm}
\usepackage{bm}
\newtheorem{proof}{Proof}
\newtheorem{corollary}{Corollary}
\usepackage{bbm}
\usepackage{pifont}
\usepackage[table]{xcolor}
\usepackage{svg}

\begin{document}
\setlength{\textfloatsep}{0.11cm}
\setlength{\baselineskip}{0.42cm}
\setlength{\abovecaptionskip}{1mm}

\title{{Constrained Capacity Analysis for Faster-than-Nyquist Signaling}}
\author{Zichao~Zhang,~\IEEEmembership{Student Member,~IEEE,}
		Melda~Yuksel,~\IEEEmembership{Senior Member,~IEEE,}
  Gokhan M. Guvensen,~\IEEEmembership{Member,~IEEE,} 
Halim~Yanikomeroglu,~\IEEEmembership{Fellow,~IEEE}
\thanks{This work was funded in part by the Scientific and Technological Research Council of Turkey, TUBITAK, under grant 122E248, and in part by a Discovery Grant awarded by the Natural Sciences and Engineering Research Council of Canada (NSERC).}
\thanks{Z. Zhang and H. Yanikomeroglu are with the Department of Systems and Computer Engineering at Carleton University, Ottawa, ON, K1S 5B6, Canada e-mail:	zichaozhang@cmail.carleton.ca, halim@sce.carleton.ca.}
\thanks{M. Yuksel and G. Guvensen are with the Department of Electrical and Electronics Engineering, Middle East Technical University, Ankara, 06800, Turkey, e-mail: {ymelda, guvensen}@metu.edu.tr.}
}

\maketitle

\begin{abstract}
 This paper studies the constrained-capacity for precoded faster-than-Nyquist (FTN) signaling with finite-alphabet inputs. Despite the promise of accelerated transmission, the fundamental rate limit of precoded FTN signaling
under practical finite-alphabet constraints remains unclear. By introducing cyclic prefix (CP) and cyclic suffix (CS), the FTN channel is decomposed into a set of parallel eigenchannels by the discrete Fourier transform (DFT) matrix, based on which the constrained capacity is derived. The results demonstrate that time acceleration can improve spectral efficiency over Nyquist signaling even when a fixed modulation order is employed. Moreover, in the low and moderate signal-to-noise ratio (SNR) regimes, a smaller constellation combined with stronger time acceleration can outperform a larger constellation with weaker acceleration. Next, the asymptotic behavior of the constrained capacity is analyzed as the acceleration factor tends to zero under both fixed transmit-SNR and fixed receive-SNR definitions. It is shown that the constrained capacity for DFT-precoded FTN is fundamentally limited by the constellation size. In addition, the constrained capacity under channel mismatch is studied and a mismatched achievable information rate (AIR) formulation is developed to show the effects of practical constraints on the performance degradation. Finally, adaptive bit loading across eigenchannels is investigated to exploit the higher-quality eigenchannels.
\end{abstract}


\section{Introduction}

The rapid growth of modern communication demands is unprecedented. As the number of connected devices continues to increase, radio spectrum is becoming an increasingly scarce resource. Early vision documents for sixth-generation (6G) wireless systems describe future service requirements in terms of immersive communication, global broadband, omnipresent IoT, spatio-temporal services, critical services, and compute-AI services \cite{whatshould6Gbe}. These use cases imply a dramatic increase in traffic volume. Consequently, new transmission techniques with higher spectral efficiency are required so that limited spectrum resources can support substantially higher communication rates.

Faster-than-Nyquist (FTN) signaling has emerged as a promising technique for improving spectral efficiency by increasing the transmission rate without expanding the occupied bandwidth. Interest in FTN signaling can be traced back to Mazo's seminal work \cite{mazo}, which showed that symbols can be transmitted faster than the Nyquist rate without compromising the error rate performance. Specifically, FTN increases the symbol rate by accelerating the symbol interval. If $T$ denotes the Nyquist symbol interval, then FTN sends symbols every $\delta T$ seconds, where $\delta \in (0,1]$ is the acceleration factor. 
A key advantage of FTN is that this rate increase is achieved without requiring additional bandwidth, since the pulse shape itself is unchanged and only the signaling rate is increased. Under the same average transmit power constraint, FTN therefore offers the potential for higher spectral efficiency. However, this gain comes at the cost of intentionally introduced intersymbol interference (ISI). Since the Nyquist interval $T$ is originally chosen to satisfy the zero-ISI criterion at the sampling instants, accelerating the symbol interval violates this condition. As a result, FTN signal detection requires more attention.

To reduce the detection complexity caused by FTN-induced ISI, transmit precoding has been widely studied as a practical solution \cite{ebrahimlowcomplx, Ebrahimlowcomplxclass, pinarlowcomplx, yu2017low, 1_R1, property,svd,precodeftnfschnl,eigenvaluedecompftn,timelocalization}. The main idea is to map information symbols onto frequency-domain components and apply precoding at the transmitter, so that after FTN transmission and receiver-side equalization, the system can be interpreted as a set of parallel subchannels. This representation greatly simplifies the FTN model and provides an appealing framework for practical implementation, and is commonly referred to as transmit precoding for FTN \cite{1_R1}. Among the candidate precoding methods, DFT-based precoding is particularly attractive because, with the aid of cyclic prefix (CP) and cyclic suffix (CS), the DFT matrix can diagonalize the FTN ISI matrix \cite{asynmacftn}. Compared with channel-dependent decompositions, the DFT has a fixed structure, low implementation complexity, and a close connection to OFDM, whose hardware and signal-processing architecture are already mature and widely available. Motivated by these advantages, this paper focuses on the constrained capacity of DFT-precoded FTN signaling with finite-alphabet inputs.

In \cite{property}, the authors investigated the eigenvalue distribution of the FTN ISI matrix and derived a capacity expression through eigenvalue decomposition.  Motivated by this perspective, \cite{svd} proposed an FTN transmission scheme based on singular value decomposition (SVD), where the channel is decomposed into parallel subchannels through SVD. Later, in \cite{precodeftnfschnl}, the same idea was extended to FTN transmission over frequency-selective channels, while \cite{eigenvaluedecompftn} applied a similar eigenvalue-decomposition framework to FTN with index modulation. In addition, \cite{timelocalization} proposed a precoding method based on the inverse square root of the FTN ISI matrix. This work also showed that precoded FTN can be optimal and capacity-achieving under Gaussian signaling. {In \cite{ofdmftn}, the authors studied an adaptive bit and power allocation scheme for multicarrier FTN. Similarly, in \cite{ftnadaptivetimedomain}, an adaptive modulation and coding scheme in time domain for FTN transmission was suggested.} Precoded FTN was also studied in other works such as \cite{precftnotfs,linearprecftn,precequalmultipathfad,zhang2026pushinglimitsunlockingpotential}. These studies provide important theoretical insights into precoded FTN signaling, but their capacity results are primarily derived under Gaussian-input assumptions. 

Although Gaussian input assumption is analytically convenient and capacity-achieving for unconstrained channels, it does not accurately represent practical communication systems, where transmitted symbols are drawn from finite-alphabet constellations such as BPSK, QPSK, and QAM. The effect of finite-alphabet constellations is captured by the constrained capacity, which characterizes the maximum achievable mutual information when the channel input is restricted to a prescribed discrete constellation rather than an arbitrary continuous Gaussian distribution. An efficient Monte Carlo method for constrained capacity evaluation was introduced by Ungerboeck \cite{ungerboeckconstrained}.

In the context of FTN signaling, constrained-capacity analysis was considered in \cite{rusek}, where bounds were derived for the case in which the transmitted time-domain symbols are drawn from a finite alphabet. The analysis assumes direct transmission of the information symbols without precoding, and the impact of ISI is addressed at the receiver through equalization. The asymptotic behavior of FTN constrained capacity when $\delta$ approaches 0 for BPSK modulation was further investigated in \cite{asympftn}. However, the paper assumes ideal and perfect equalization at the receiver and assumes sinc pulses, which is also impractical. Due to the large state space FTN induces, receiver equalization is generally highly complex. Instead, transmitter-side precoding can offer a lower-complexity implementation. Motivated by this observation, this paper investigates the constrained capacity of precoded FTN signaling, with a particular focus on DFT-based precoding. By employing CP and CS, the FTN channel can be transformed into a circulant form, allowing the DFT matrix to diagonalize the channel into independent eigenchannels and enabling low-complexity detection. 

In practice, however, the CP and CS introduce transmission overhead and may be omitted or shortened. In this case, the FTN channel is no longer perfectly circulant, and DFT precoding cannot completely diagonalize the channel, resulting in residual inter-eigenchannel interference. This transmitter-receiver incompatibility is commonly referred to as channel mismatch. The theoretical foundation of mismatched decoding and generalized mutual information was established in \cite{airmis}, while related capacity results for mismatched channels were further studied in \cite{capmis}. These tools were applied to evaluate the AIR of various practical decoders in \cite{rusekair,Gokhanmismatch,airbicm}. Following this line of work, in this paper, we also quantify the resulting performance degradation due to channel mismatch, and develop a mismatched achievable information rate (AIR) formulation for mismatched DFT-precoded FTN signaling.




Our main contributions are summarized as follows:
\begin{itemize}
    \item We derive the constrained capacity of FTN signaling with transmitter-side precoding, where finite-alphabet symbols are transmitted.
    \item We analyze the asymptotic behavior of the constrained capacity of DFT-precoded FTN signaling as the acceleration factor $\delta$ approaches $0$ under both fixed transmit-SNR and fixed receive-SNR definitions.
    \item We derive the constrained capacity under a mismatched channel model, which captures the practical degradation caused by imperfect channel diagonalization.
    \item We further investigate adaptive bit loading for DFT-precoded FTN signaling in order to overcome the rate limitation imposed by finite-alphabet inputs.
\end{itemize}
Overall, these contributions demonstrate that DFT-precoded FTN signaling is a promising technology for future communication standards, enabling low complexity FTN implementation.

For notational convenience, the following symbols are used throughout this paper. The superscripts $*$, $T$, and $\dagger$ denote complex conjugate, transpose, and conjugate transpose, respectively. The symbol $\star$ denotes convolution. The notation $(\cdot)_{m,n}$ represents the $(m,n)$th entry of a matrix.  The terms $y(n\delta T)$ and $y[n]$ are used interchangeably to denote the $n$th time-domain sample of the signal $y(t)$.  The operators $\mathrm{tr}(\cdot)$ and $\mathbb{E}[\cdot]$ denote the trace of a matrix and the expectation. Finally, $\bm{I}_N$ represents the $N\times N$ identity matrix.

\section{System Model}\label{sec:systemmodel}

\begin{figure}
    \centering
    \includegraphics[width=1\linewidth]{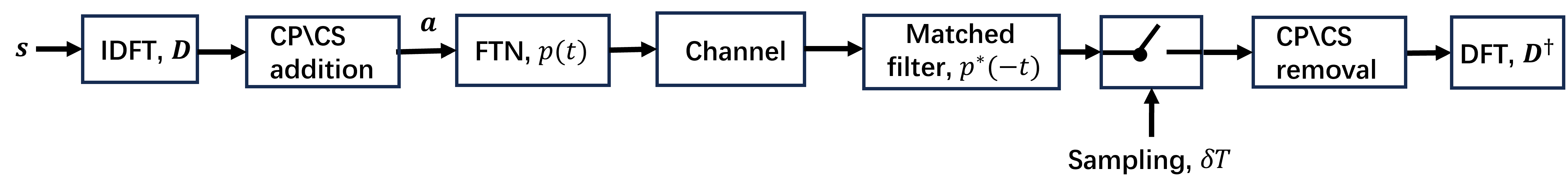}
    \caption{Block diagram of DFT-precoded FTN transmission system. }
    \label{fig:sysmodeldftprecftn}
\end{figure}
Assume that $N$ symbols are transmitted. The independent and identically distributed (i.i.d.) data symbols are denoted by $\bm{s}=[s_0,s_1,\dots,s_{N-1}]^T$, where each symbol is drawn from a finite alphabet $\mathcal{S}$ of size $M$. Typical examples include QPSK with $M=4$ and 16QAM with $M=16$. Throughout this paper, we assume unit-energy constellation symbols, i.e.,
\begin{equation}
    \mathbb{E}[|s_i|^2]=1, \label{eqn:dataunitenj}
\end{equation}
Let $E_s$ denote the allocated symbol energy. Under uniform power allocation, the frequency-domain symbol vector is precoded by the DFT matrix before pulse shaping as shown in Fig.~\ref{fig:sysmodeldftprecftn}.
Let $\bm{D}$ denote the $N\times N$ normalized DFT matrix, whose entries are given by
\begin{equation}
    (\bm{D})_{k,n}=\frac{1}{\sqrt{N}}e^{-j\frac{2\pi}{N}kn},
    \qquad k,n=0,\dots,N-1. \label{eqn:dftmatterms}
\end{equation}
The precoded transmit vector is then written as
\begin{equation}
    \bm{a}=\sqrt{E_s}\bm{D}\bm{s}.
\end{equation} The precoded symbols $a_n$, $n=0,\dots,N-1$, are then pulse-shaped by $p(t)$ and transmitted at the accelerated symbol interval $\delta T$, where $T$ denotes the Nyquist symbol period and $\delta \in (0,1]$ is the acceleration factor. The derivations in this paper only require the pulse-shaping filter $p(t)$ to be of unit-energy and band-limited with bandwidth $W$. For instance, the pulse-shaping filter $p(t)$ can be a root-raised cosine (RRC) pulse with a roll-off factor $\beta$. Moreover, we have $\bm{\Sigma}_a=\mathbb{E}[\bm{a}\bm{a}^\dagger]=E_s\bm{D}\mathbb{E}[\bm{s}\bm{s}^\dagger]\bm{D}^\dagger=E_s\bm{D}\bm{D}^\dagger$. 

For simplicity, we consider an additive white Gaussian noise (AWGN) channel. At the receiver, the matched filter $p^*(-t)$ is employed. The output of the matched filter is defined as
\begin{equation}
    g(t)=p(t)\star p^*(-t).
\end{equation}
Then, the continuous-time input-output relationship can be written as
\begin{align}
    {y}(t)
    &= \sum_{n=0}^{N-1} a_n g(t-n\delta T) + n(t) \\
    &= x(t) + n(t),
\end{align}
where ${y}(t)$ is the received signal, $x(t)$ denotes the transmitted signal after pulse shaping, and $n(t)$ is AWGN with power spectral density $\sigma_0^2$. After matched filtering, the filtered noise process is denoted by
    $\eta(t)=n(t)\star p^*(-t)$.
The matched-filter output is sampled every $\delta T$ seconds. Therefore, the discrete-time system model can be expressed as
\begin{equation}
    {\bm{y}}=\bm{G}\bm{a}+\bm{\eta},
\end{equation}
where
\begin{align}
    {\bm{y}} &= [{y}(0),\, {y}(\delta T),\, \dots,\, {y}((N-1)\delta T)]^T, \\
    \bm{a} &= [a_0,\, a_1,\, \dots,\, a_{N-1}]^T, \\
    \bm{\eta} &= [\eta(0),\, \eta(\delta T),\, \dots,\, \eta((N-1)\delta T)]^T,
\end{align}
and $\bm{G}$ is the ISI matrix whose entries are given by
    \begin{equation}
        (\bm{G})_{i,j}=g((i-j)\delta T), \qquad i,j=0,\dots,N-1.\label{eqn:gmatterms}
    \end{equation} 
It follows that $\bm{G}$ is an $N\times N$ Hermitian Toeplitz matrix. Moreover, the noise vector $\bm{\eta}$ is correlated, with the covariance matrix
\begin{equation}
    \mathbb{E}[\bm{\eta}\bm{\eta}^\dagger]=\sigma_0^2\bm{G}.
\end{equation}

\subsection{DFT-Precoded FTN}

In practice, the composite pulse response becomes negligible beyond a finite number of samples. We therefore assume that there exists an integer $L$ such that
\begin{equation}
    g[n]=0, \qquad |n|>L.
\end{equation}
Since the DFT matrix is able to perfectly diagonalize a cyclic matrix \cite{gray},  CP and CS of length $L$ are appended to $\bm{a}$ before pulse shaping, yielding
\begin{equation}
    [
\underbrace{a_{N-L}, a_{N-L+1}, \ldots, a_{N-1}}_{\mathrm{CP}},
a_0,\ldots,a_{N-1},
\underbrace{a_0,\ldots,a_{L-1}}_{\mathrm{CS}}
].
\end{equation}
After matched filtering and sampling, a total of $N+2L$ samples are obtained at the receiver, and the first $L$ and last $L$ samples are discarded.
Due to this operation, the equivalent channel model can be expressed as
\begin{equation}
    \bm{y}_c=\bm{G}_c\bm{a}+\bm{\eta},
\end{equation}
where $\bm{G}_c$ is an $N\times N$ circulant matrix. Its first row is given by
\begin{equation}
    [g(0),\, g(-\delta T),\, \dots,\, g(-L\delta T),\, 0,\, \dots,\, 0,\, g(L\delta T),\, \dots,\, g(\delta T)].
\end{equation}
Since $\bm{G}_c$ is circulant, it can be diagonalized by the DFT matrix \cite{gray} as
\begin{equation}
    \bm{G}_c=\bm{D}\bm{\Lambda}\bm{D}^\dagger,
    \label{eqn:Gcdecomp}
\end{equation}
where $\bm{\Lambda}=\mathrm{diag}\{\lambda_0,\lambda_1,\dots,\lambda_{N-1}\}$ contains the eigenvalues of $\bm{G}_c$.

At the receiver, the sampled matched-filter output $\bm{y}$ is multiplied by $\bm{D}^\dagger$ so that the channel is represented in the frequency domain. Define
\begin{equation}
    \tilde{\bm{y}}\triangleq[\tilde{y}[0],\tilde{y}[1],\dots,\tilde{y}[N-1]]^T.
\end{equation}
Then,
\begin{align}
    \tilde{\bm{y}}
    &=\bm{D}^\dagger\bm{y}_c \notag\\
    &=\bm{D}^\dagger\bm{G}_c\bm{a}+\bm{D}^\dagger\bm{\eta} \notag\\
    &=\sqrt{E_s}\bm{D}^\dagger\bm{G}_c\bm{D}\bm{s}+\bm{D}^\dagger\bm{\eta} \notag\\
    &\overset{(a)}{=}\sqrt{E_s}\bm{\Lambda}\bm{s}+\bm{\omega},
    \label{eqn:commchnldecomp}
\end{align}
where (a) follows from \eqref{eqn:Gcdecomp}, and 
    $\bm{\omega}\triangleq\bm{D}^\dagger\bm{\eta}$.
The covariance matrix of $\bm{\omega}$ is
    $\mathbb{E}[\bm{\omega}\bm{\omega}^\dagger]
    =\bm{D}^\dagger\mathbb{E}[\bm{\eta}\bm{\eta}^\dagger]\bm{D}
    =\sigma_0^2\bm{\Lambda}$.
Therefore, the components $\omega_i$ are independent, and the FTN channel is decomposed into $N$ parallel eigenchannels. The $i$th eigenchannel is given by
\begin{equation}
    \tilde{y}[i]=\sqrt{E_s}\lambda_i s_i+\omega_i,
    \qquad i=0,\dots,N-1,
    \label{eqn:eigenchnlexpression}
\end{equation}
where
    $\omega_i\sim\mathcal{CN}(0,\lambda_i\sigma_0^2)$.

As discussed in \cite{property}, the eigenvalues $\lambda_i$ can be approximated by the samples of the folded spectrum $G_d(f_n)$ over one period. 
The
folded spectrum is the discrete time Fourier transform of the sampled pulse autocorrelation
\(g[\ell]=g(\ell\delta T)\), and can be expressed as
\begin{align}
    G_d(f_n)&=\sum_{\ell=-L}^{L}g[\ell]e^{j2\pi f_n\ell} \\
    &=\frac{1}{\delta T}\sum_{k=-\infty}^{+\infty}G\left(\frac{f_n-k}{\delta T}\right),    \quad f_n\in \left[0,1\right],
\end{align}
where $G(f)$ is the continuous-time Fourier transform of $g(t)$. As an example, we plot one period of the folded-spectrum in Fig.~\ref{fig:eigenval} for an RRC pulse. The folded spectrum is periodic with period 1. The length of the support of the folded spectrum is $\max(\delta TW, 1)$ \cite{zhang2022faster}.
Specifically,
\begin{equation}
    \lambda_\ell \approx \frac{1}{\delta T}\sum_{k=-\infty}^{+\infty}
    G\!\left(\frac{\ell}{N\delta T}-\frac{k}{\delta T}\right)
    \triangleq G_d\left(\frac{\ell}{N}\right).
    \label{eqn:lambsampleG}
\end{equation}
The fraction of nonzero eigenvalues is approximately determined by the support of $G_d(f_n)$ over one period \cite{property}.
\begin{figure}
    \centering
    \includegraphics[width=0.8\linewidth]{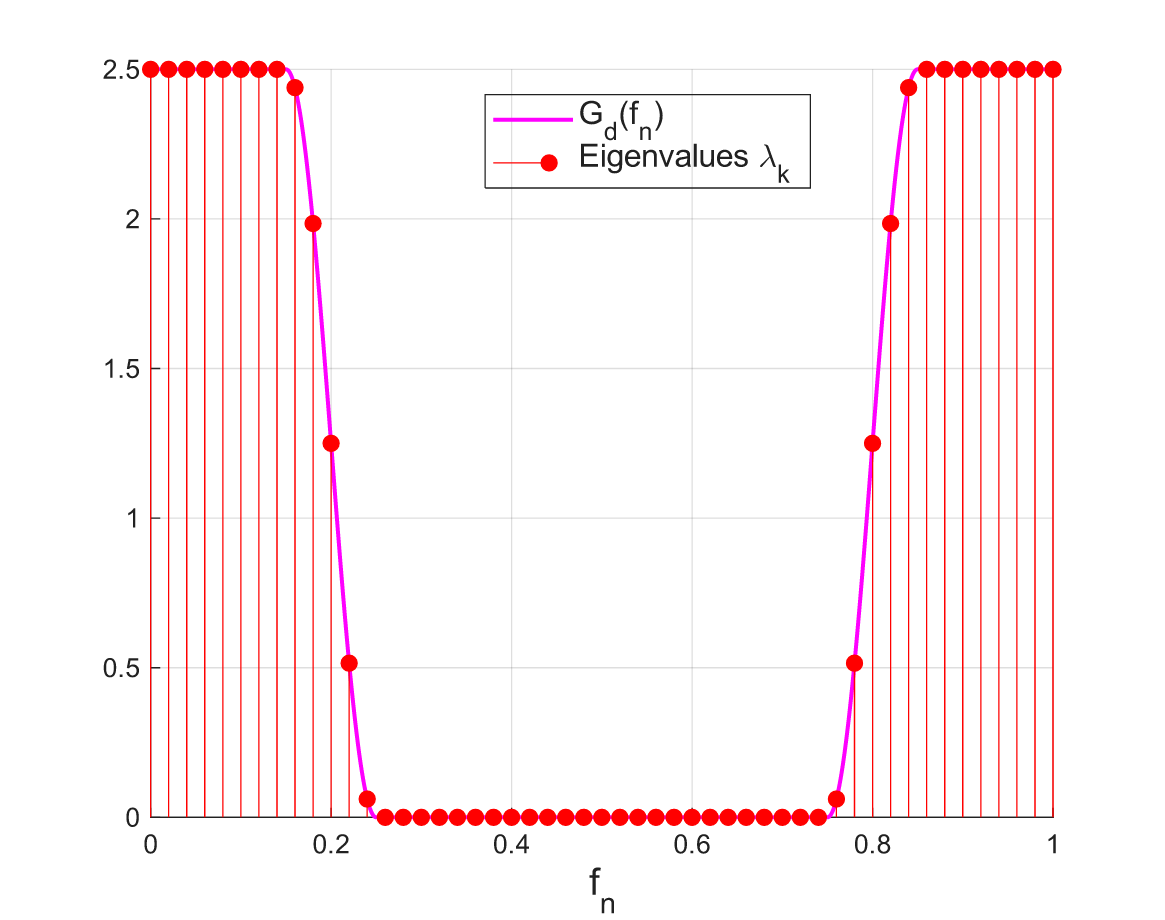}
    \caption{The folded spectrum $G_d(f_n)$ of an RRC pulse with $\delta=0.4$ and $\beta=0.25$ over one normalized frequency period of $[0,1)$. The corresponding eigenvalues $\lambda_i$ of the circulant FTN ISI matrix $\bm{G}_c$ are the samples of $G_d(f_n)$.}
    \label{fig:eigenval}
\end{figure}
Let $N_r$ denote the number of nonzero eigenvalues. As can be seen in Fig.~\ref{fig:eigenval}, when $\delta TW<1$, there will be zero samples or zero eigenvalues over one period and for sufficiently large $N$, we can approximate the fraction
\begin{equation}
    \frac{N_r}{N}\approx \min\bigl(1,\delta TW\bigr).
    \label{eqn:NandNred}
\end{equation}
Hence, when $\delta TW<1$, only $N_r$ eigenchannels are active, while the remaining $N-N_{{r}}$ eigenchannels are effectively nulled by zero eigenvalues.  
Therefore, we transmit symbols only on the active eigenchannels, which belong to the index set 
\begin{align}
    \mathcal{D}=\left\{0, \dots, \frac{N_r}{2},  N-\frac{N_r}{2}, \dots, N-1\right\}. \label{eqn:symmodulateindex}
\end{align}

\subsection{Power Constraint and SNR Definitions}\label{sec:powtxrx}

We next investigate the power constraint used for the
DFT-precoded FTN system. 
For FTN signaling with transmit power $P_{TX}$, sending $N$ symbols over a duration of $N\delta T$ yields a total energy assumption of $NP_{TX}\delta T$, so the average energy per symbol is $E_s=P_{TX}\delta T$. Therefore, for fixed power, decreasing $\delta$ reduces the energy allocated to each symbol, which results in a shorter minimum Euclidean distance and hence a higher error probability for practical constellations. To fairly evaluate FTN performance under different $\delta$, two SNR definitions are required: the transmit SNR, $\mathsf{SNR_{tx}}=\frac{P_{TX}}{\sigma_0^2}$, and the received SNR, $\mathsf{SNR_{rx}}=\frac{E_s/T}{\sigma_0^2}=\frac{P_{TX}\delta}{\sigma_0^2}$. These two definitions coincide only for Nyquist signaling $(\delta=1)$, whereas if $\delta\neq1$ they must be distinguished \cite{zhang2025iapr}. 

Assume that the physical transmit-power
limit is \(P_{ TX}\). For sufficiently large \(N\), the overhead due
to the CP and CS is neglected in the following
power analysis. The average transmit-power constraint can then be written as
\begin{align}
    P_{TX}
    &=\mathbb{E}\left[\frac{1}{N\delta T}\int_{-\infty}^{\infty}|x(t)|^2dt\right] \notag\\
    &=\frac{1}{N\delta T}\mathrm{tr}(\bm{G}_c\bm{\Sigma}_a),\notag\\
    &\overset{(a)}{=}
    \frac{E_s}{N\delta T}\mathrm{tr}(\bm{\Lambda}),
    \label{eqn:powconst}
\end{align}
where (a) follows from \eqref{eqn:Gcdecomp} and \eqref{eqn:dataunitenj}.

To satisfy the power constraint, the symbol energy $E_s$ must satisfy
\begin{equation}
    E_s=\frac{NP_{TX}\delta T}{\sum_{d'=0}^{N_r-1}\lambda_{d'}}. \label{eqn:symbolematch}
\end{equation}
To avoid
confusion with the global time-domain SNR definitions, the effective frequency-domain SNR of
the \(d\)th active eigenchannel is denoted by \(\gamma_d\). Substituting \eqref{eqn:symbolematch} into \eqref{eqn:eigenchnlexpression} gives
\begin{equation}
    \gamma_d
    =\frac{\lambda_d^2E_s}{\lambda_d\sigma_0^2}
    =\frac{\lambda_d NP_{TX}\delta T}
    {\sigma_0^2\sum_{d'=0}^{N_r-1}\lambda_{d'}}.
    \label{eqn:snrichannelabthres}
\end{equation}
Furthermore, when $\delta TW<1$, using \eqref{eqn:NandNred}, \eqref{eqn:snrichannelabthres} can be rewritten as
\begin{equation}
    \gamma_{d}
    =\frac{\lambda_d N_rP_{TX}}
    {\sigma_0^2W\sum_{d'=0}^{N_r-1}\lambda_{d'}}. \label{eqn:snrichannelblthres}
\end{equation}
\begin{remark}\label{remark:txrxsnr}
    For the rest of the paper, conclusions derived under fixed $\mathsf{SNR_{tx}}$ can be converted to fixed $\mathsf{SNR_{rx}}$ condition by simply replacing $P_{TX}$ with $\delta P_{TX}$.
\end{remark}

\section{Constrained Capacity Derivation} \label{sec:conscapacityderivation}

  Different from the Gaussian capacity, constrained capacity assumes that the channel input is restricted to a finite
constellation set \(\mathcal{S}\). It
therefore characterizes the maximum achievable information rate with ideal channel coding for practical systems. In other words, it gives the information-theoretic limit
that coded modulation schemes can
approach. In this section, we derive the constrained capacity of the DFT-precoded FTN
system



Using the constrained capacity result for the scalar AWGN channel given in \cite{ungerboeckconstrained}, 
the constrained capacity of the $d$th eigenchannel in \eqref{eqn:eigenchnlexpression} 
can be written as
\begin{align}
    &C_d(\delta)
    =
    \log_2 M
    -\frac{1}{M}\sum_{m=0}^{M-1}
    \int_{-\infty}^{+\infty}
    \frac{\exp\!\left(-\frac{\omega^2}{2\lambda_d\sigma_0^2}\right)}
    {2\pi\lambda_d\sigma_0^2}
    \notag\\
    & \times
    \log_2\!\left(
    \sum_{\ell=0}^{M-1}
    \exp\!\left(
    -\frac{
    |u-\lambda_d\sqrt{E_s}(s[m]-s[\ell])|^2-u^2
    }{2\lambda_d\sigma_0^2}
    \right)
    \right)du.
    \label{eqn:itheigenchnlcap}
\end{align}
The constrained capacity of SISO FTN with transmitter-side precoding is then obtained by summing the information carried over all active eigenchannels and averaging it over the transmitted time-domain symbols, namely,
\begin{equation}
    C(\delta)=\frac{1}{N}\sum_{d=0}^{N_{{r}}-1} C_d(\delta).
    \label{eqn:persymcap}
\end{equation}
Note that \eqref{eqn:persymcap} is measured in bits per transmitted symbol. 
Then, the constrained capacity in bits/s/Hz is given by
\begin{equation}
    C'(\delta)=\frac{1}{\delta TW}C(\delta).
    \label{eqn:capbpsphz}
\end{equation}
Substituting \eqref{eqn:NandNred} into \eqref{eqn:capbpsphz}, we obtain the following theorem.
{\begin{theorem}
\label{thm:conscap_dft_ftn}
For the DFT-precoded FTN system with CP/CS and finite-alphabet input 
$s_d\in\mathcal{S}$, $|\mathcal{S}|=M$, the constrained capacity in bits/s/Hz is
\begin{equation}
    C'(\delta)
    =
    \frac{1}{N_{{r}}}
    \sum_{d=0}^{N_{{r}}-1} C_d(\delta),
    \label{eqn:aveovereigenchnl}
\end{equation}
where $C_d(\delta)$ is the constrained capacity of the $d$th active eigenchannel given in \eqref{eqn:itheigenchnlcap}. 
\end{theorem}}
\begin{remark}\label{rem:cpcsrem}
    For sufficiently large \(N\), the overhead factor \(N/(N+2L)\) approaches one when \(L\) is fixed
relative to \(N\). If finite-block overhead is explicitly counted, the
final spectral efficiency expressions in this section should be
multiplied by \(N/(N+2L)\).
\end{remark}


Theorem~\ref{thm:conscap_dft_ftn} is also consistent with the $2WT$ theorem discussed in \cite{kim2026fasterthannyquistsignalingfinitetimebandwidth}, which states that packing more FTN symbols does not necessarily increase the signaling dimensions, the number of effective signaling dimensions is determined by the occupied time bandwidth product $N\delta WT$, which is exactly equal to $N_r$. Theorem~\ref{thm:conscap_dft_ftn} conveys the same insight: the FTN block is decomposed into $N$ frequency-domain eigenchannels. However, when $\delta(1+\beta)<1$, only about $N_r$ of them have nonzero eigenvalues and can effectively carry information, although $N$ time-domain FTN samples are transmitted. 
\begin{corollary}\label{thm:asympftntx}
    Under fixed $\mathsf{SNR_{tx}}$, the constrained capacity of DFT-precoded FTN given in Theorem~\ref{thm:conscap_dft_ftn} approaches a constant value as $\delta\to0$. 
\end{corollary}
\begin{proof}
    See Appendix~\ref{app:asympconscaptxsnr} for the proof.
\end{proof}
\begin{corollary}\label{thm:asymprxsnrconscap}
     Under fixed $\mathsf{SNR_{rx}}$, the constrained capacity of DFT-precoded FTN approaches $\log_2M$ as $\delta\to0$.
\end{corollary}
\begin{proof}
    See Appendix~\ref{app:app2} for the proof.
\end{proof}
\begin{remark}
    As $\delta$ decreases, symbols are packed more densely in time. However, under fixed $\mathsf{SNR_{tx}}$, the transmit power $P_{TX}$ is kept fixed, and thus the energy allocated to each time-domain symbol decreases with $\delta$. Therefore, the potential gain from higher symbol rate is suppressed by the reduced symbol energy. As can be observed from \eqref{eqn:asympsnrtx}, the effective SNR of each active eigenchannel remains finite. Consequently, the constrained capacity approaches a constant as $\delta$ approaches $0$. In contrast, under fixed $\mathsf{SNR_{rx}}$, the received symbol energy is kept the same as $\delta$ decreases, which is equivalent to allowing the power level $P_{TX}/\delta$ to grow unbounded. In this case, the effective SNR of the active eigenchannels increases without bound, and the rate is eventually limited only by the finite constellation size.
\end{remark}

\section{Achievable Information Rate with Channel Mismatch}\label{sec:diffpowallo}

Although the introduction of CP and CS enables DFT matrix to decompose the FTN channel to independent parallel eigenchannels, they reduce spectral efficiency.
When no CP or CS is used, the finite-block FTN matrix $\bm{G}$ is only Toeplitz rather than circulant. Therefore, the DFT matrix does not diagonalize the matrix $\bm{G}$. Consequently,
residual inter-eigenchannel interference remains after receiver processing. In this case, the receiver still treats the
subchannels as if they are independent, although the actual channel
contains residual interference which is correlated among eigenchannels. Therefore, the 
performance measure is no longer the constrained capacity of the ideal
matched eigenchannel model. Instead, it is the AIR induced by the mismatched decoding metric. This mismatched AIR quantifies the reliable achievable rate as if a very good channel code exists. Thus, 
it provides a direct measure of the performance degradation
caused by imperfect DFT diagonalization in a coded FTN system.

Under receiver mismatch, the transformed channel matrix is
\begin{equation}
\mathbf{\Gamma} \triangleq \mathbf{D}\mathbf{G}\mathbf{D}^{H}.\label{eqn:gammamat}
\end{equation}
Since $\mathbf{G}$ is Toeplitz, the $(k,\ell)$th entry of $\mathbf{\Gamma}$ is
\begin{equation}
\Gamma_{k,\ell}
=
\sum_{m=0}^{N-1}\sum_{n=0}^{N-1}
(\mathbf{D})_{k,m}\,(\bm{G})_{m,n}\,(\mathbf{D}^{H})_{n,\ell}. \label{eqn:termgammamat}
\end{equation}
Substituting \eqref{eqn:dftmatterms} and \eqref{eqn:gmatterms} into \eqref{eqn:termgammamat} gives
\begin{equation}
\Gamma_{k,\ell}
=
\frac{1}{N}
\sum_{m=0}^{N-1}\sum_{n=0}^{N-1}
g[m-n]\,
e^{-j\frac{2\pi}{N}km}\,
e^{j\frac{2\pi}{N}\ell n}.
\end{equation}
With a change of variables $r=m-n$, we have
\begin{equation}
\Gamma_{k,\ell}
=
\frac{1}{N}
\sum_{r=-(N-1)}^{N-1}
g[r]\,
e^{-j\frac{2\pi}{N}kr}
\sum_{n=n_0(r)}^{n_1(r)}
e^{-j\frac{2\pi}{N}(k-\ell)n},
\label{eq:Akl_exact}
\end{equation}
where $n_0(r)=\max(0,-r)$ and $n_1(r)=\min(N-1,N-1-r)$ For $k\neq \ell$, the inner summation is a finite geometric series. Hence, using the identity $\sum_{n=a}^{b} q^n
=
q^a\frac{1-q^{b-a+1}}{1-q},$
we have


\begin{align}
\Gamma_{k,\ell}
&= \frac{1}{N}
\sum_{r=-(N-1)}^{N-1}
g[r] e^{-j\frac{2\pi}{N}kr}
e^{-j\frac{2\pi}{N}(k-\ell)n_0(r)}
\notag\\
&\quad \times
\frac{
1-e^{-j\frac{2\pi}{N}(k-\ell)(N-|r|)}
}{
1-e^{-j\frac{2\pi}{N}(k-\ell)}
},
\quad k\neq \ell ,
\label{eq:offdiag_exact}
\end{align}
 which represent the residual inter-eigenchannel interference after DFT-domain diagonalization. 

Under mismatched decoding, we do not employ eigenchannels with small gains as they have weak    SNR gains. Therefore, we still modulate the frequency-domain symbols on the $N_r$ strong eigenchannels as discussed in Section~\ref{sec:conscapacityderivation} and we still use the index set as defined in \eqref{eqn:symmodulateindex}. 



\subsection{Equivalent Interference Model}

Under the mismatched model, the input-output relation is
\begin{equation}
\tilde{\bm{y}} = \bm{D}^\dagger\bm{y}=\sqrt{E_s}\bm{D}^\dagger\bm{G}\bm{D}\bm{s}+\bm{D}^\dagger\bm{\eta}=\sqrt{E_s}~\mathbf{\Gamma}\mathbf{s} + \bm{\omega}, \label{eqn:truemismatchmodel}
\end{equation}
 with the $d$th eigenchannel  written as 
\begin{align}
\tilde{y}_d &= \sqrt{E_s}~\Gamma_{d,d}s_d + \sqrt{E_s}\sum_{\substack{\ell \neq d\\ \ell \in \mathcal{D}}} \Gamma_{d,\ell}s_\ell + \omega_d \notag \\
&=\sqrt{E_s}~\Gamma_{d,d}s_d + \xi_d + \omega_d\label{eqn:pereigenchnlinoutputwithinterference}
\end{align}
Here the noise term $\bm{\omega}$ has the covariance $\mathbb{E}[\bm{\omega}\bm{\omega}^\dagger]=\sigma_0^2\bm{D}^\dagger\bm{G}\bm{D}=\sigma_0^2\bm{\Gamma}$. Also note that $\mathbb{E}[|\omega_d|^2]=\sigma_0^2\Gamma_{d,d}$.
Similar to \eqref{eqn:powconst}, the energy allocated to each eigenchannel $E_s$ under the mismatched model should satisfy the power constraint 
\begin{align}
\frac{E_s}{N\delta T}\text{tr}\left(\bm{G}\bm{D}\mathbb{E}[\bm{s}\bm{s}^\dagger]\bm{D}^\dagger\right) 
    &=\frac{E_s}{N\delta T}\text{tr}\left(\bm{\Gamma}\right)=P_{TX},
\end{align}
and we have 
\begin{align}
    E_s=\frac{NP_{TX}\delta T}{\sum_{d=0}^{N_r-1}\Gamma_{d,d}}.
\end{align}
Since the symbols $\{s_\ell\}$ are independent, zero mean, and have unit energy, the interference variance of the $d$th eigenchannel is
\begin{equation}
\sigma_{I,d}^2
\triangleq
\mathbb{E}\left[
\left|\sqrt{E_s}
\sum_{\substack{\ell \neq d\\ \ell \in \mathcal{D}}} \Gamma_{d,\ell}s_\ell
\right|^2
\right]=E_s\sum_{\substack{\ell \neq d\\ \ell \in \mathcal{D}}}|\Gamma_{d,\ell}|^2 .
\end{equation}

Next, we model the residual interference $\xi_d$ as complex Gaussian since Gaussian noise is the worse case for a given variance \cite{modelintasgausnoise}. 
 Then, the eigenchannel model becomes
\begin{equation}
\tilde{y}_d \approx \sqrt{E_s}~\Gamma_{d,d} s_d + \tilde{\omega}_d, \label{eqn:mismatchmodel}
\end{equation}
where
\begin{equation}
    \tilde \omega_d \sim \mathcal{CN}(0,\nu_d^2),
    \qquad
    \nu_d^2
    \triangleq
    \sigma_0^2 \Gamma_{d,d}
    +
    E_s\sum_{\substack{\ell \neq d\\ \ell \in \mathcal{D}}}|\Gamma_{d,\ell}|^2 .
    \label{eqn:misnoisevar}
\end{equation}
As a result, the effective $\mathsf{SNR_{tx}}$ of the $d$th eigenchannel is
\begin{equation}
SNR_d
=
\frac{|\Gamma_{d,d}|^2 E_s}{\nu_d^2}.
\label{eq:effective_snr_general}
\end{equation}

\subsection{Mismatched AIR derivation}
For finite-alphabet inputs with uniform distribution, the standard
mismatched AIR is obtained by averaging the mismatched information
density under the true observation law \cite{mismatchorig,Gokhanmismatch}. Specifically, we let $q_d(\tilde y_d|c)$ denote the assumed likelihood assigned to symbol $c\in\mathcal S$ on the $d$th eigenchannel. Then the corresponding mismatched information density for the $d$th eigenchannel where $\bar{c}$ is the actual symbol being sent is
\begin{equation}
    \iota_d(s_d,\tilde y_d)
    =
    \log_2 M
    -
    \log_2
    \left(
    \frac{\sum\limits_{c\in\mathcal S} q_d(\tilde y_d|s_d=c)}
    {q_d(\tilde y_d|s_d=\bar{c})}
    \right)
    \label{eq:mm_branch_info_density}
\end{equation}
in bits/symbol. The quantity \(\imath_d({s}_d,\tilde{y}_d)\) measures how strongly the
mismatched receiver output \(\tilde{y}_d\) supports the actually
transmitted symbol relative to all possible constellation hypotheses.
The denominator is the assumed likelihood of the true transmitted
symbol, while the numerator is the total assumed likelihood over all
candidate symbols. If the likelihoods assigned to all incorrect symbols
are close to zero and \(q_d(\tilde{y}_d|\bar{c})\) dominates the sum,
then
\[
\sum_{c\in\mathcal{S}}q_d(\tilde{y}_d|c)
\approx
q_d(\tilde{y}_d|\bar{c}),
\]
and the information density approaches \(\log_2 M\).  This corresponds to
a highly reliable observation, where the receiver can correctly identify the
transmitted constellation point. In contrast, if the likelihoods of
several incorrect symbols are comparable to the likelihood of the true
symbol, then the ratio inside the logarithm becomes large and the
information density decreases. In the extreme case, where all constellation
points are almost equally likely under the mismatched metric, the ratio is
approximately \(M\), and the information density approaches zero.
The mismatched AIR in bits per transmitted time-domain symbol is given by
\begin{equation}
    I_{\rm mm}(\delta)
    =
    \mathbb{E}_{\bm{s}}
    \left[
    \frac{1}{N_r}
    \sum_{d=0}^{N_r-1}
    \iota_d(s_d,\tilde y_d)
    \right],
    \label{eq:mm_air_symbol}
\end{equation}
where the expectation averages over all transmitted symbol vectors and noise realizations generated by the actual FTN channel model in \eqref{eqn:truemismatchmodel}. 
We compute the likelihood $q_d(\tilde y_d|s_d=c)$ inside the expectation according to the approximated model \eqref{eqn:mismatchmodel} as 
\begin{equation}
    q_d(\tilde y_d|s_d=c)
    =
    \frac{1}{\pi \nu_d^2}
    \exp\!\left(
    -\frac{|\tilde y_d-\sqrt{E_s}\,\Gamma_{d,d}c|^2}{\nu_d^2}
    \right),
    \quad a\in\mathcal S.
    \label{eq:mm_branch_metric}
\end{equation} 
This mismatch between the true observation law and the assumed decoding rule quantifies the practical degradation caused by imperfect diagonalization.
\begin{theorem}
    The mismatched AIR for DFT-precoded FTN without CP/CS in bits/s/Hz is
\begin{align}
    I_{\rm mm}'(\delta)
    &=
    \frac{1}{\delta TW}\,I_{\rm mm}(\delta),
    \label{eq:mm_air_bpshz}
\end{align}
where $I_{\rm mm}$ is given by \eqref{eq:mm_air_symbol}. 
\end{theorem}
\begin{corollary}
    When $\delta TW<1$, and $N_{ r}/N \approx \delta TW$, then \eqref{eq:mm_air_bpshz} can be approximated as
\begin{equation}
    I_{\rm mm}'(\delta)
    \approx
    \mathbb{E}_{\bm{s}}
    \left[
    \frac{1}{N_r}
    \sum_{d=0}^{N_r-1}
    \imath_d(s_d,\tilde y_d)
    \right].
    \label{eq:mm_air_bpshz_approx}
\end{equation}
\end{corollary}
\begin{remark}
    In \eqref{eq:mm_air_bpshz_approx}, the expectation is taken under the true correlated FTN law. Thus, a closed-form evaluation is generally not available. The mismatched AIR is computed numerically by Monte Carlo simulation.
\end{remark}

\section{Adaptive Bit Loading}
\label{sec:adaptivebitload}



The DFT-precoded FTN system with CP and CS decomposes the channel into a set of parallel eigenchannels, where the $d$th eigenchannel gain is $\lambda_d$. Each eigenchannel is generally not equally strong.  This is also visible in Fig.~\ref{fig:eigenval} for the matched case.
A similar argument of reliability also appears in the mismatched case. 
Eigenchannels with stronger gains and weaker interference are more reliable than others.  

As shown in Corollaries~\ref{thm:asympftntx} and~\ref{thm:asymprxsnrconscap},
for a fixed finite-alphabet constellation of size $M$, the maximum amount of achievable information rate for each eigenchannel can carry is limited by  $\log_2M$.  Unless a larger constellation size is employed, the rate can not increase further. Hence, when all eigenchannels are forced to use the same modulation order, the stronger eigenchannels cannot fully exploit their favorable channel conditions, while the weaker eigenchannels may not be able to reliably support that constellation. 

This observation motivates adaptive bit loading across the FTN eigenchannels. Instead of assigning a universal constellation to all eigenchannels, we enable the modulation order to vary with the eigenchannel quality. This strategy is known to provide a larger overall spectral efficiency than fixed modulation in Nyquist systems \cite{goldsmith}. 

 
In this work, we focus on square $M$-QAM constellations due to their practicality.  For a square $M$-QAM constellation, the error probability $P_e$ can be approximately calculated as 
\begin{align}
    P_e\approx1-\left[1-2\left(1-\frac{1}{\sqrt{M}}Q\left(\sqrt{\frac{3\gamma}{M-1}}\right)\right)\right]^2,\label{eqn:mqampe}
\end{align}
where $Q(\cdot)$ is the Q function and $\gamma$ is the $\mathsf{SNR_{tx}}$.

\begin{figure*}
\begin{equation}
P_{e,d}=
\begin{cases}
1-\left[1-2\left(1-\frac{1}{\sqrt{M}}Q\left(\sqrt{\frac{3\lambda_d NP_{TX}\delta T}
    {(M-1)\sigma_0^2\sum_{d'=0}^{N-1}\lambda_{d'}}}\right)\right)\right]^2, & \delta TW \geq 1,\\
1-\left[1-2\left(1-\frac{1}{\sqrt{M}}Q\left(\sqrt{\frac{3\lambda_d N_rP_{TX}T}
    {(M-1)\sigma_0^2(1+\beta)\sum_{d'=0}^{N-1}\lambda_{d'}}}\right)\right)\right]^2, & \delta TW<1.
\end{cases} \label{eqn:errprob}
\end{equation}
\end{figure*}

For the matched case, inserting \eqref{eqn:snrichannelabthres} and \eqref{eqn:snrichannelblthres} into \eqref{eqn:mqampe}, we obtain the error probability $P_{e,d}$ for the $d$th eigenchannel. It is shown in \eqref{eqn:errprob} at the top of the next page. 
The modulation order for the $d$th eigenchannel is then chosen according to a target symbol error rate threshold $P_{e}^{\mathrm{th}}$. Specifically, for each eigenchannel, we select the largest modulation order $M_d$ from the set $\mathcal{M}$,  $\mathcal{M}=\{4,16,64,256\}$, such that 
\begin{equation}
M_d = \max \left\{ M_d\in \mathcal{M} : P_{e,d} \le P_e^{\mathrm{th}} \right\}.
\label{eq:abi_max_rule}
\end{equation}
If there is no modulation order that can be larger than the minimum allowable constellation that satisfies the target reliability requirement, then the eigenchannel is assigned the lowest-order constellation.

Once the modulation order $M_d$ is determined for each eigenchannel, the constrained capacity of the $d$th eigenchannel is evaluated using the same constrained-capacity formula as in \eqref{eqn:itheigenchnlcap}, but with the constellation chosen for that channel. Denoting the corresponding constrained capacity by $C_{a,d}(\delta)$, the overall achievable spectral efficiency $C_a(\delta)$ under adaptive bit loading is
\begin{equation}
C_a(\delta)=\frac{1}{N}\sum_{i=0}^{N_r-1} C_{a,d}(\delta),
\label{eq:abl_rate_symbol}
\end{equation}
in bits per transmitted symbol, and
\begin{equation}
C_a'(\delta)=\frac{1}{\delta TW}\, C_a(\delta)
\label{eq:abl_rate_bshz}
\end{equation}
in bits/s/Hz. When $\delta TW<1$ and $N_{{r}}/N \approx \delta TW$, this can also be written as
\begin{equation}
C_a'(\delta)\approx \frac{1}{N_r}\sum_{d=0}^{N_{r-1}} C_{a,d}(\delta).
\label{eq:abl_rate_avg}
\end{equation}

For the mismatched case, the eigenchannel quality is determined jointly by the diagonal gain and the residual inter-eigenchannel interference. Accordingly, adaptive bit loading assigns higher-order constellations to eigenchannels with larger effective SNRs as in \eqref{eq:effective_snr_general}, while weaker eigenchannels use lower-order modulation. 
The benefit of adaptive bit loading is that it removes the unnecessary ceiling imposed by fixed modulation.   As a result, the total achievable spectral efficiency can exceed that of any single fixed-modulation scheme, especially in high SNR.

\section{Simulation Results}




In this section, we present numerical results for DFT-precoded
FTN signaling with finite-alphabet inputs. We set the pulse shaping filter to be RRC with roll-off factor $\beta=0.25$. We first evaluate the constrained capacity under the
matched parallel eigenchannel model derived in Section~III. We then compute the mismatched AIR based on
the derivations in Section~IV. Finally, we investigate adaptive
bit loading across the eigenchannels according to the
eigenchannel quality for the matched and mismatched cases based on Section~V. 
 The results are presented in bits/s/Hz.  In the simulations, we use both fixed-$\mathsf{SNR_{tx}}$ and fixed-$\mathsf{SNR_{rx}}$ power profiles.  

 \begin{figure}
    \centering
    \includegraphics[width=0.8\linewidth]{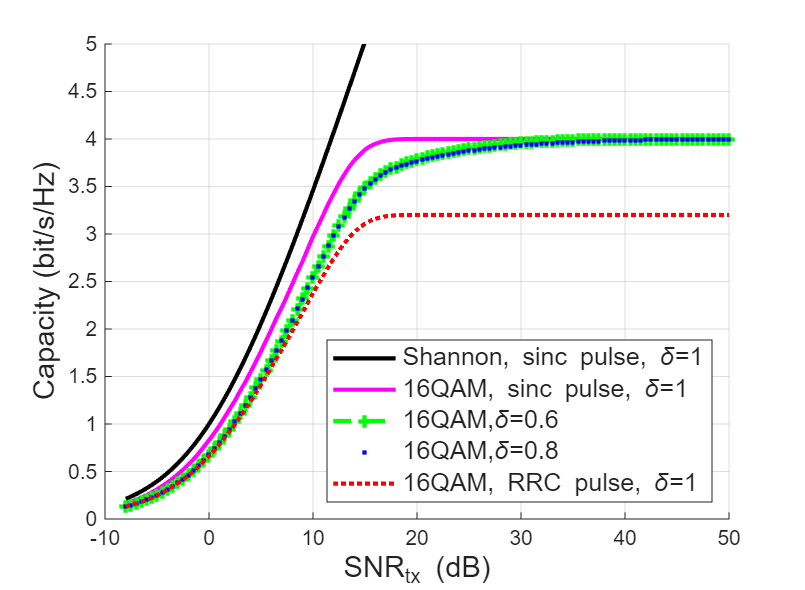}
    \caption{ Constrained capacity vs $\mathsf{SNR_{tx}}$ for different $\delta$. For curves not using sinc pulse, the RRC roll-off factor is $\beta=0.25$.}
    \label{fig:txcapvssnr}
\end{figure}

\begin{figure}
    \centering
    \includegraphics[width=0.8\linewidth]{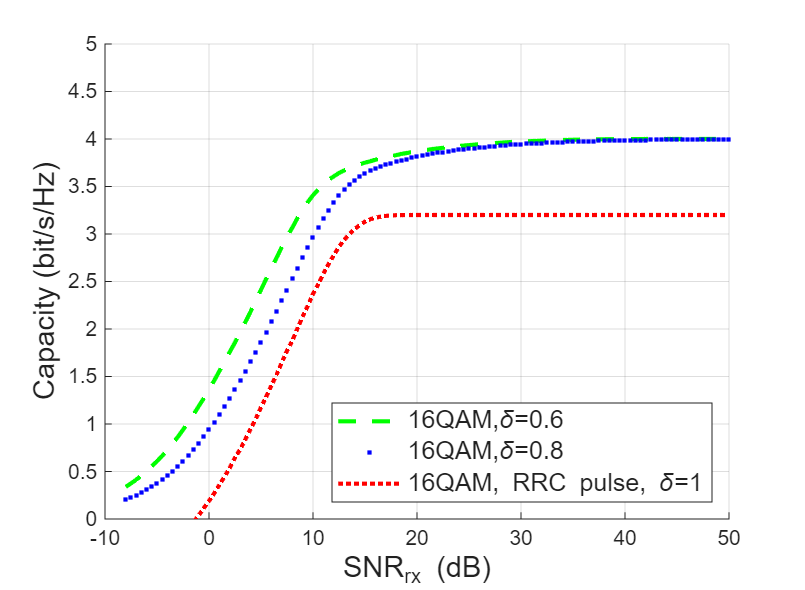}
    \caption{Constrained capacity vs $\mathsf{SNR_{rx}}$ for different $\delta$. For the curves, the RRC roll-off factor is $\beta=0.25$.}
    \label{fig:rxcapvssnr}
\end{figure}

Figs.~\ref{fig:txcapvssnr} and~\ref{fig:rxcapvssnr} demonstrate the constrained capacity of DFT-precoded FTN for fixed $\mathsf{SNR_{tx}}$ and $\mathsf{SNR_{rx}}$ respectively. We can see that FTN increases the spectral efficiency compared to Nyquist transmission when RRC pulse with $\beta=0.25$ is used for both fixed $\mathsf{SNR_{tx}}$ and $\mathsf{SNR_{rx}}$ cases. The curves for $\delta=0.6$ and $\delta=0.8$ overlap in Fig.~\ref{fig:txcapvssnr} since under fixed $\mathsf{SNR_{tx}}$, decreasing $\delta$ does not further improve the capacity as is discussed in Corollary~\ref{thm:asympftntx}. The 16QAM sinc pulse curve with $\delta=1$ is included only as a Nyquist benchmark, whereas the Shannon capacity curve serves as an upper bound. 
At high $\mathsf{SNR_{tx}}$, the FTN curves with
$\delta=0.6$ and $\delta=0.8$ approach the information rate upper bound 
$\log_2 M$ bits/s/Hz. 
This behavior follows from \eqref{eqn:aveovereigenchnl}, when $\delta TW<1$, the spectral efficiency is
the average of constrained capacity over the active eigenchannels.
At sufficiently high SNR, each active eigenchannel approaches
its finite-alphabet limit of $\log_2 16=4$ bits per symbol.
Consequently, the overall spectral efficiency approaches
$4$ bits/s/Hz, which is the same high-SNR limit as Nyquist 16QAM
with sinc pulse. Therefore, FTN can recover the spectral efficiency
loss caused by the RRC.
In Fig.~\ref{fig:rxcapvssnr}, under fixed $\mathsf{SNR_{rx}}$, because of the increasing physical transmission power, for smaller $\delta$, a better performance can be achieved. At high
$\mathsf{SNR_{rx}}$, both curves converge because the
constrained capacity is limited by the ceiling of
$\log_2 16$ bits/s/Hz.

\begin{figure}
    \centering
    \includegraphics[width=0.8\linewidth]{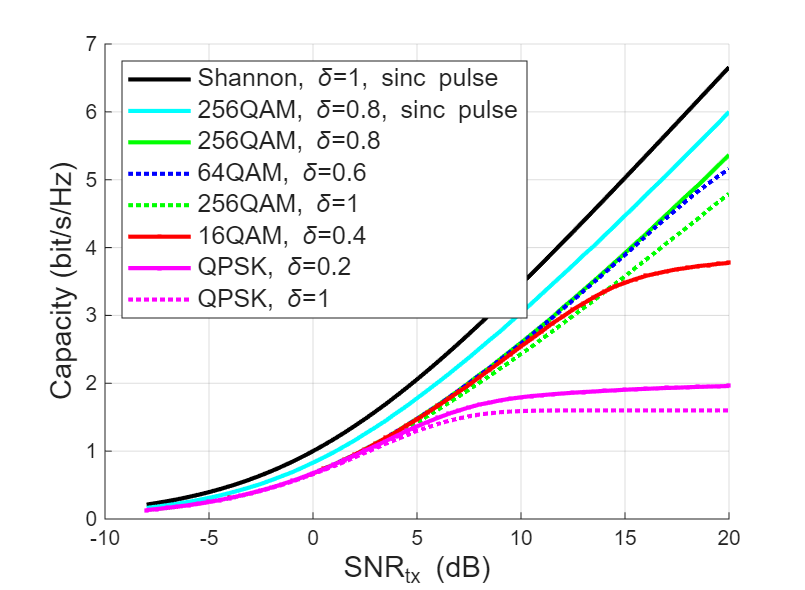}
    \caption{Constrained capacity vs $\mathsf{SNR_{tx}}$, where the spectral efficiency $\log_2M/\delta$ in bits/s is fixed for all curves except the curve for QPSK with $\delta=1$ and 256QAM with $\delta=1$. For curves not using sinc pulse, the RRC roll-off factor is $\beta=0.25$.}
    \label{fig:logmdivtautx}
\end{figure}

Fig.~\ref{fig:logmdivtautx} shows the constrained capacity under fixed $\mathsf{SNR_{tx}}$ for several $(M,\delta)$ pairs where the ratio $\frac{\log_2M}{\delta}$ is kept the same across different $(M,\delta)$ pairs. In the low and moderate SNR
regions, a smaller acceleration factor allows a lower order
constellation to achieve nearly the same constrained capacity as a higher-order constellation with a larger
$\delta$. Specifically, the curves for QPSK with $\delta=0.2$
and 16QAM with $\delta=0.4$ are nearly identical up to
approximately $5$ dB. Similarly, 16QAM with $\delta=0.4$
performs comparably to 64QAM with $\delta=0.6$ up to
approximately $10$ dB, while 64QAM with $\delta=0.6$ and
256QAM with $\delta=0.8$ exhibit similar performance up to
approximately $17.5$ dB. These results show that FTN with higher symbol rate can allow a reduction in
constellation order. This is
practically attractive because lower-order constellations require
lower detection complexity and can provide a more favorable
instantaneous-to-average power ratio.

\begin{figure}
    \centering
    \includegraphics[width=0.8\linewidth]{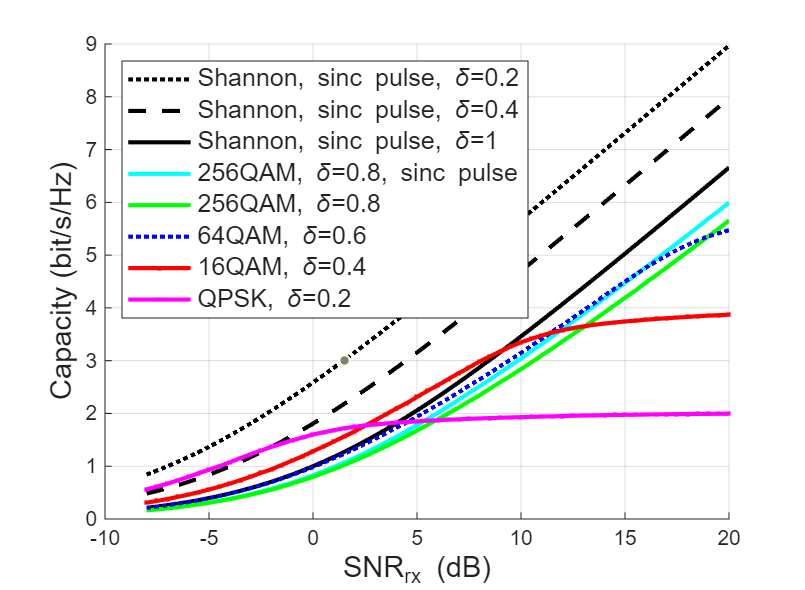}
    \caption{Constrained capacity vs $\mathsf{SNR_{rx}}$, where the spectral efficiency $\log_2M/\delta$ in bits/s is fixed for all curves. For curves not using sinc pulse, the roll-off factor is $\beta=0.25$. The Shannon capacity curves for $\delta=0.2, 0.4$ are plotted as upper bounds, where the corresponding $\mathsf{SNR_{rx}}$ is computed as $P_{TX}/\delta$.  }
    \label{fig:logmdivtaurx}
\end{figure}

Fig.~\ref{fig:logmdivtaurx} shows the results when $\mathsf{SNR_{rx}}$ is fixed. In the low and moderate SNR scenarios, smaller constellation size combined with smaller acceleration factor outperforms larger constellations with larger $\delta$ value pairs. When finite alphabet is used, a larger constellation does not automatically guarantee a larger capacity, because the receiver must still reliably distinguish among more points.  Consequently, more aggressive FTN packing with a smaller alphabet becomes more preferable in the low-$\mathsf{SNR_{rx}}$ region. As $\mathsf{SNR_{rx}}$ continues to increase, the larger constellations gradually recover their advantage, because the symbol decisions become sufficiently reliable and the achievable rate approaches the alphabet entropy ceiling. 

\begin{figure}
    \centering
    \includegraphics[width=0.8\linewidth]{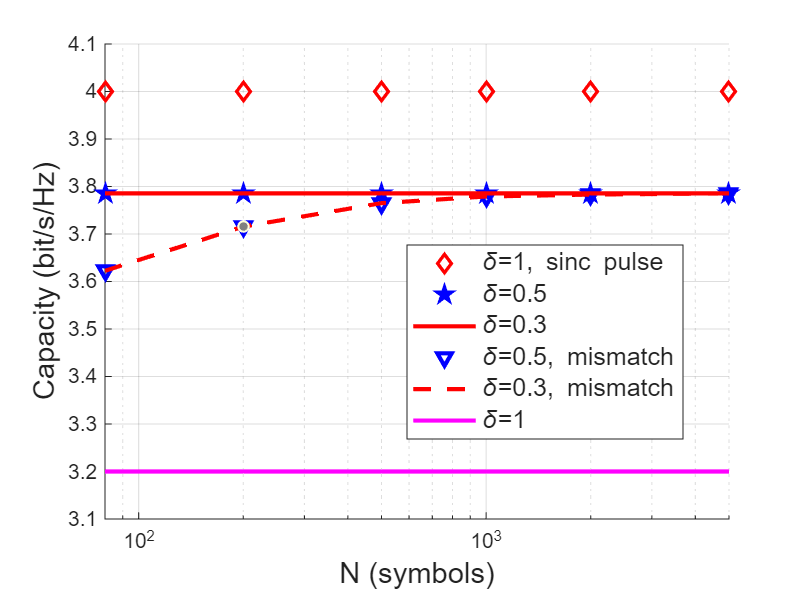}
    \caption{Constrained capacity and mismatched AIR vs. blocklength $N$ for different acceleration factors $\delta$ for 16QAM, where $\mathsf{SNR_{tx}}=$20 dB. For curves not using sinc pulse, the RRC roll-off factor is $\beta=0.25$.}
    \label{fig:mistx}
\end{figure}

\begin{figure}
    \centering
    \includegraphics[width=0.8\linewidth]{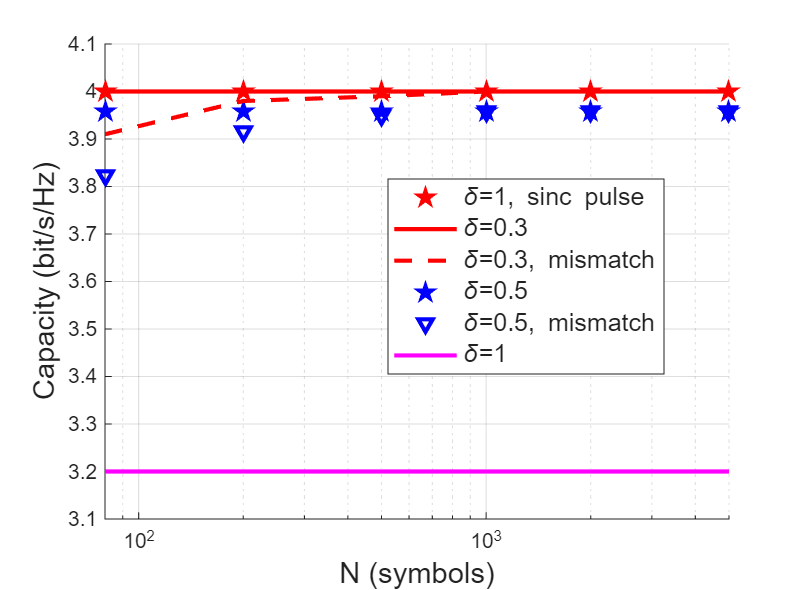}
    \caption{Constrained capacity and mismatched AIR vs. blocklength $N$ for different acceleration factors $\delta$ for 16QAM, where $\mathsf{SNR_{rx}}=$20 dB. For curves not using sinc pulse, the RRC roll-off factor is $\beta=0.25$.}
    \label{fig:misrx}
\end{figure}

We now turn to the performance degradation caused by channel mismatch. The mismatched AIR is evaluated according to \eqref{eq:mm_branch_info_density}--\eqref{eq:mm_branch_metric}. Specifically, the likelihood term $q_d(\tilde{y}_d|s_d)$ inside the expectation is computed using \eqref{eq:mm_branch_metric}, while $\tilde{y}_d$ is generated from the correlated channel model in \eqref{eqn:truemismatchmodel}. The expectation is then approximated by Monte Carlo simulation over the generated received vector $\bm{y}$, where the noise term $\bm{\omega}$ is correlated.
In Fig.~\ref{fig:mistx}, we compare, for fixed $\mathsf{SNR_{tx}}$, the constrained capacity without mismatch derived in \eqref{eqn:aveovereigenchnl} in Section~\ref{sec:conscapacityderivation} and the mismatched AIR in \eqref{eq:mm_air_bpshz_approx} in Section~\ref{sec:diffpowallo} for different block lengths $N$. 
For finite
$N$, imperfect DFT diagonalization leaves residual inter-eigenchannel
interference, causing the mismatched AIR to fall below the matched
constrained capacity. As $N$ increases, the DFT asymptotically
diagonalizes the Toeplitz channel matrix \cite{gray}, so the interference
weakens and the gap narrows for every $\delta$.
{It is also observed that the curves for \(\delta=0.5\) and \(\delta=0.3\) overlap in Fig.~\ref{fig:mistx} for fixed $\mathsf{SNR_{tx}}$. Both acceleration factors satisfy \(\delta(1+\beta)<1\) with $\beta=0.25$. In this regime, only the active eigenchannels are used for transmission, with \(N_r/N\approx \delta(1+\beta)\), and the spectral efficiency in bit/s/Hz is determined by the average constrained capacity over the active eigenchannels. Therefore, further reducing \(\delta\) from \(0.5\) to \(0.3\) does not change the constrained capacity under the fixed \(\mathsf{SNR}_{ tx}\) setting. The corresponding mismatched AIR curves also nearly coincide, since the finite-block residual interference follows the same active-eigenchannel structure and vanishes as \(N\) becomes large.} 
Fig.~\ref{fig:misrx} presents the corresponding results for fixed $\mathsf{SNR_{rx}}$. As $N$ increases, the residual
inter-eigenchannel interference vanishes, and the mismatched AIR
approaches the matched constrained capacity for each $\delta$.
This convergence, however, does not necessarily imply convergence
to the finite-alphabet upper bound. Since $\mathsf{SNR_{ rx}}$ is
fixed, the smaller acceleration
factor $\delta=0.3$ provides higher effective SNRs on the active
eigenchannels, allowing the 16QAM constrained capacity to approach
its upper bound of 4 bits/s/Hz. For $\delta=0.5$, the
power increase is smaller, and some weaker eigenchannels have not
yet reached their finite-alphabet saturation level at
$\mathsf{SNR_{rx}}=20$~dB. Hence, its asymptotic mismatched AIR remains below the
upper bound.

\begin{figure}
    \centering
    \includegraphics[width=0.8\linewidth]{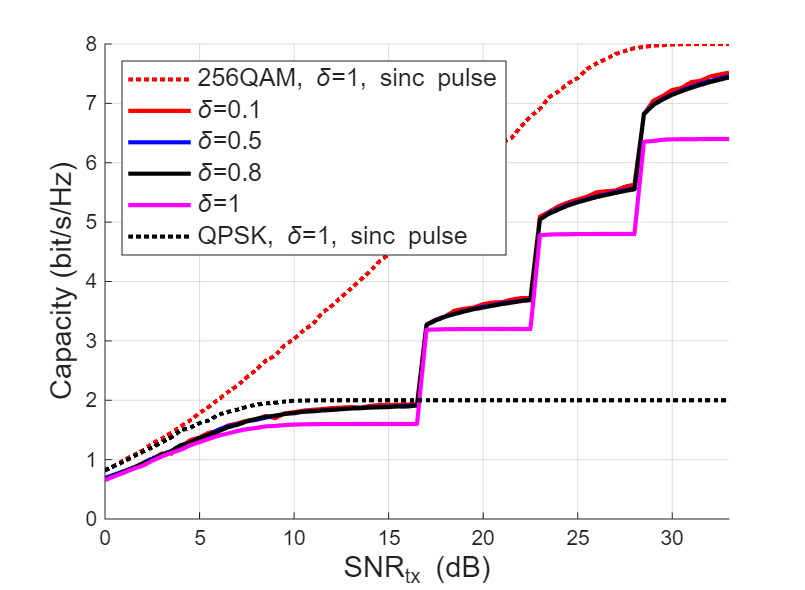}
    \caption{Constrained capacity with adaptive bit loading vs $\mathsf{SNR_{tx}}$, where RRC pulse with $\beta=0.25$ is applied. }
    \label{fig:adapsnrtx}
\end{figure}

We now look into the simulation results of the adaptive bit loading scheme in Section~\ref{sec:adaptivebitload} for both matched and mismatched cases.  
We show in Figs.~\ref{fig:adapsnrtx} and~\ref{fig:adapsnrrx} the constrained capacity under fixed $\mathsf{SNR_{tx}}$ and $\mathsf{SNR_{rx}}$ with adaptive bit loading across the FTN eigenchannels under a target error-probability threshold of $10^{-3}$ with blocklength $N=200$. We choose the set of modulation levels as $\mathcal{M}=\{4,16,64,256\}$ for the matched case. For each eigenchannel, the corresponding $\mathsf{SNR_{tx}}$ is computed from its eigenvalue, and the assigned modulation order is selected as the highest one whose error probability remains below $10^{-3}$. Thus, bit loading is performed independently according to the channel quality of each eigenchannel rather than using a single modulation order for the whole system. Stronger eigenchannels can support higher-order constellations, while weaker eigenchannels use lower-order modulation or carry no information.

The staircase behavior in Fig.~\ref{fig:adapsnrtx} is caused by the discrete modulation selection. The total rate changes only when the $\mathsf{SNR_{tx}}$ of one or more eigenchannels becomes sufficient to satisfy the $10^{-3}$ error-probability requirement for the next modulation level. Between two transition points, the loading pattern gradually shifts to higher-order constellations as $\mathsf{SNR_{tx}}$ increases, and we can see a gradual increase in the constrained capacity.  At $\delta=0.8$, 
the folded spectrum is nonuniform, producing stronger
eigenchannels in the plateau region and weaker eigenchannels
near the roll-off region. Adaptive bit loading assigns higher order constellations to the
stronger eigenchannels and lower order constellations to the
weaker ones, resulting in a higher spectral efficiency than
that obtained with $\delta=1$. When $\delta TW<1$,
further acceleration does not introduce additional signaling
dimensions. Instead, the effective SNR
 over the remaining active eigenchannels remains
approximately unchanged under fixed $\mathsf{SNR_{ tx}}$.
Therefore, the curves for
$\delta=0.5$ and $\delta=0.1$ are nearly identical.

\begin{figure}
    \centering
    \includegraphics[width=0.8\linewidth]{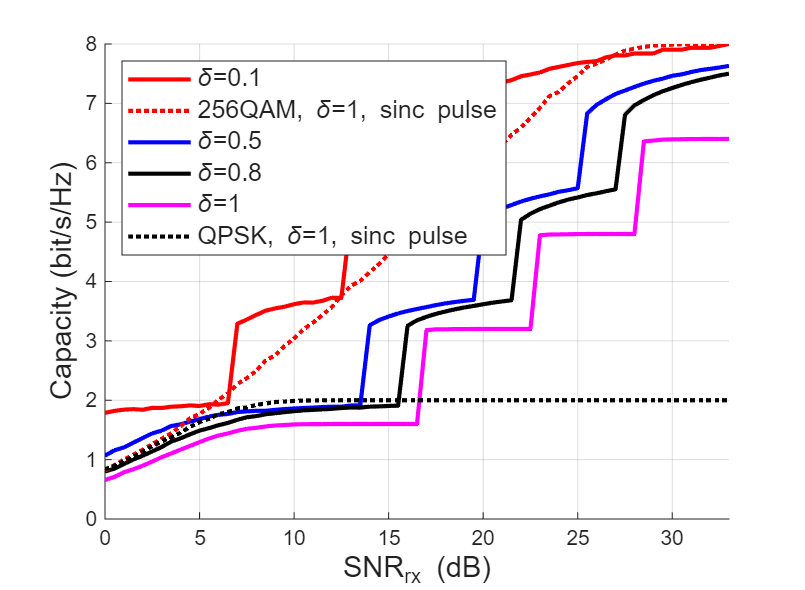}
    \caption{Constrained capacity with adaptive bit loading vs $\mathsf{SNR_{rx}}$, where RRC pulse with $\beta=0.25$ is applied.}
    \label{fig:adapsnrrx}
\end{figure}

In Fig.~\ref{fig:adapsnrrx}, the curves exhibit a staircase behavior similar to Fig.~\ref{fig:adapsnrtx}. However, under fixed $\mathsf{SNR_{rx}}$, the spectral
efficiency increases as $\delta$ decreases. This is because maintaining
the same $\mathsf{SNR_{ rx}}$ at a smaller $\delta$ requires a higher
physical transmit power, which increases the effective SNR of the
active eigenchannels. Consequently, higher-order constellations are
selected for smaller $\delta$, shifting the curve to the left.

\begin{figure}
    \centering
    \includegraphics[width=0.8\linewidth]{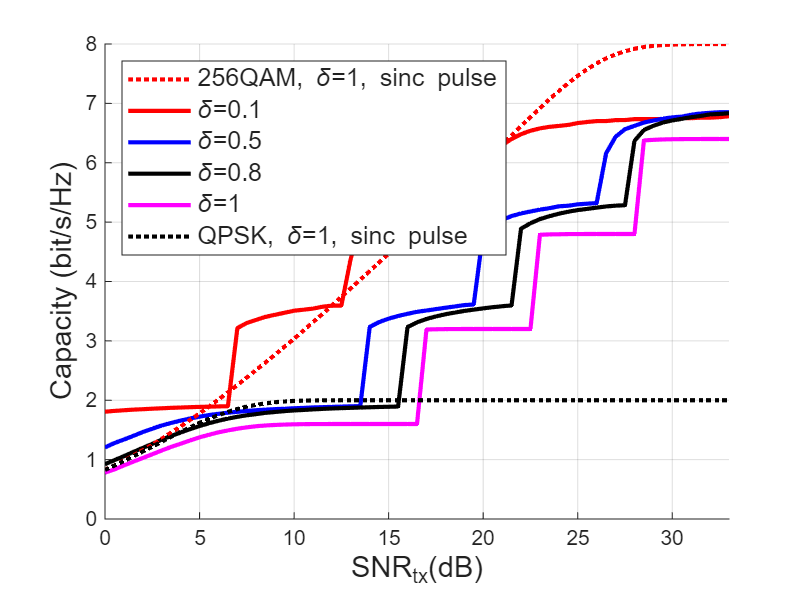}
    \caption{Mismatched AIR with adaptive bit loading vs $\mathsf{SNR_{rx}}$, where RRC pulse with $\beta=0.25$ is applied.}
    \label{fig:adaptivemismatchrx}
\end{figure}

In Fig.~\ref{fig:adaptivemismatchrx}, we show the results for mismatched AIR with adaptive bit loading under fixed $\mathsf{SNR_{rx}}$ under a target error-probability threshold of $10^{-3}$ with the same blocklength $N=200$. We choose the set of modulation levels as $\mathcal{M}=\{4,16,64,256\}$. In contrast to the matched case, the residual interference from other eigenchannels reduces the reliability of each eigenchannel and therefore lowers the achievable rate. Nevertheless, adaptive bit loading remains beneficial, since stronger eigenchannels can still support higher-order constellations.

\begin{figure}
    \centering
    \includegraphics[width=0.8\linewidth]{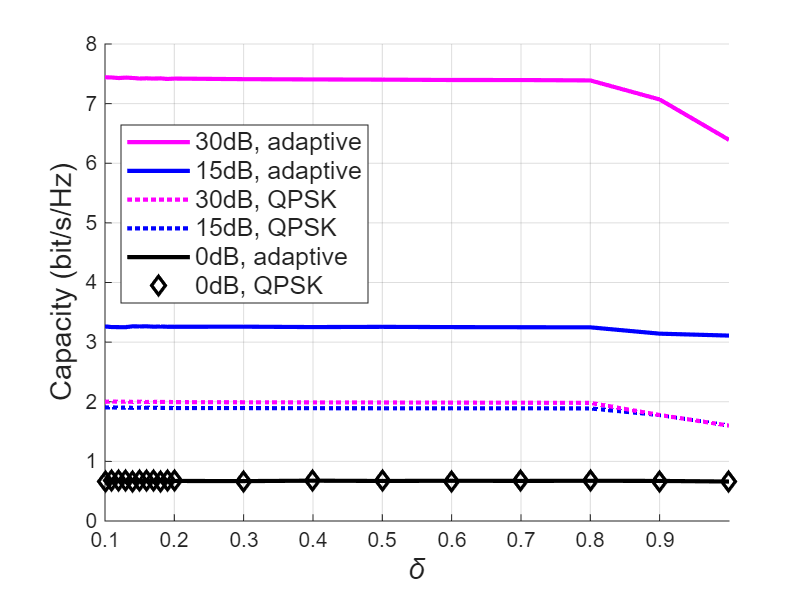}
    \caption{Constrained capacity vs. $\delta$ for fixed $\mathsf{SNR_{tx}}$, for RRC pulses with $\beta=0.25$.}
    \label{fig:adaptautx}
\end{figure}

Figs.~\ref{fig:adaptautx} and~\ref{fig:adaptaurx} show the constrained capacity versus the acceleration factor $\delta$ under adaptive bit loading with the same target error-probability threshold of $10^{-2}$ for fixed $\mathsf{SNR_{tx}}$ and $\mathsf{SNR_{rx}}$ respectively. These figures directly reveal how the loading outcome changes as the acceleration factor varies. We compare the adaptive loading results with fixed constellation of QPSK.  In Fig.~\ref{fig:adaptautx}, we can see that the constrained capacity shows similar behavior as the Gaussian capacity shown in \cite{zhang2022faster}, the constrained capacity rises as $\delta$ decreases and then saturates at the Gaussian capacity threshold $\delta=0.8$ \cite{zhang2022faster}. For sufficient $\mathsf{SNR_{tx}}$, if higher order modulation is allowed, the adaptive curve chooses higher order modulation and outperforms the fixed-alphabet result.


\begin{figure}
    \centering
    \includegraphics[width=0.8\linewidth]{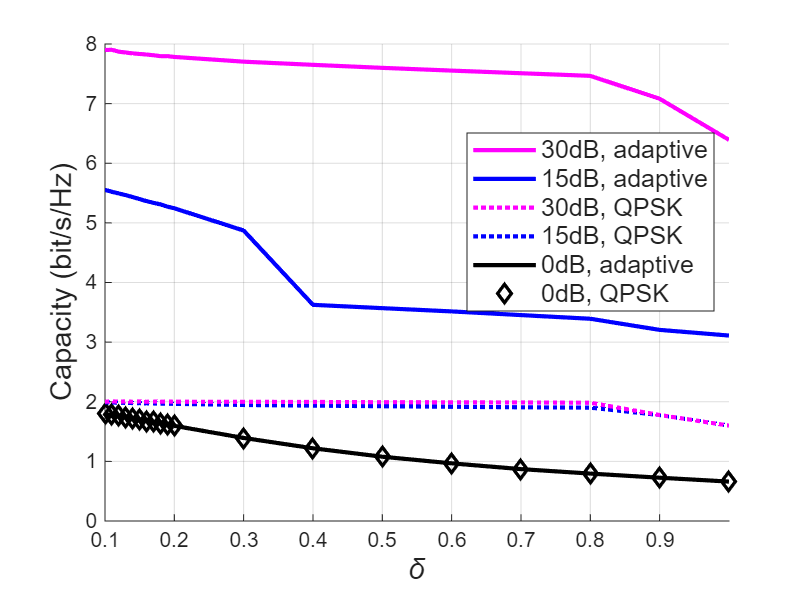}
    \caption{Constrained capacity vs. $\delta$ for fixed $\mathsf{SNR_{rx}}$, for RRC pulses with $\beta=0.25$.}
    \label{fig:adaptaurx}
\end{figure}

In Fig.~\ref{fig:adaptaurx}, the adaptive scheme outperforms fixed QPSK because
the increased transmit power associated with a smaller $\delta$ under
fixed $\mathsf{SNR_{ rx}}$ enables more eigenchannels to support
higher order modulation. However, at $\mathsf{SNR_{ rx}}=30$ dB,
the additional gain is limited because the largest available
constellation is 256QAM. Most strong eigenchannels have
already been assigned 256QAM, whose rate is bounded by
8 bits per symbol. Therefore,  decreasing
$\delta$ further provides only a small improvement, mainly through the
remaining weaker eigenchannels.

\section{Conclusion}\label{sec:conclusion}

This paper characterized the constrained capacity of DFT-precoded
FTN signaling with finite-alphabet inputs. The results show that
time acceleration can improve spectral efficiency over Nyquist
signaling even when a fixed modulation order is used, although the
achievable rate remains bounded by the constellation size. Over the
low and moderate SNR ranges, combining a smaller
constellation with a smaller acceleration factor, i.e., a higher
symbol transmission rate, can outperform a larger constellation
with a larger acceleration factor. This operating point is also attractive
from an implementation perspective because lower order constellations
generally require simpler detection and may provide more favorable
instantaneous-to-average power-ratio characteristics. Furthermore,
the rate loss caused by imperfect DFT diagonalization becomes small
as the block length increases, since the residual inter-eigenchannel
interference rapidly diminishes.

Adaptive bit loading was also shown to be feasible for both the
matched and mismatched FTN models. By assigning higher order
constellations to stronger eigenchannels and lower order
constellations to weaker ones, adaptive loading improves spectral
efficiency compared with fixed modulation while satisfying the reliability requirement. These findings show that the
acceleration factor, modulation order, block length, and receiver
model should be selected jointly in practical finite input FTN
systems. Future work will extend the constrained-capacity analysis
to equalization-based FTN receivers and compare their achievable
rates and complexity with the DFT-precoded framework considered
here.

\appendices

\section{Proof for Theorem~\ref{thm:asympftntx}}\label{app:asympconscaptxsnr}



In this appendix, we prove Theorem~\ref{thm:asympftntx} stated in Section~\ref{sec:conscapacityderivation}. 
We need to show that, under fixed
\(\mathsf{SNR_{ tx}}\), the constrained capacity does not grow
without bound as \(\delta \to 0\). The proof proceeds in three
steps. First, we use the folded-spectrum approximation to
replace the sum of active eigenvalues by its limiting integral.
Second, we show that the product \(\lambda_d E_s\), which
determines the effective SNR of the \(d\)th eigenchannel,
converges to a finite quantity. Third, we substitute this limiting
eigenchannel model into the finite-alphabet constrained-capacity
expression and average over the active eigenchannels. This
shows that the limiting spectral efficiency is a finite constant
independent of \(\delta\).
First, we define
\begin{equation}
f_d=\frac{d}{N\delta T},
\qquad
\Delta f=\frac{1}{N\delta T}.
\end{equation}
Assume $N$ is sufficiently large, approximate the sum of the eigenvalues in the denominator of \eqref{eqn:symbolematch} using the Riemann sum as  
\begin{align}
\sum_{d=0}^{N_{{r}}-1}\lambda_j
&\overset{(a)}{\approx}
\frac{1}{\delta T}\sum_{d=0}^{N_r-1}G(f_d)=N\sum_{j=0}^{N_r-1}G(f_d)\Delta f \notag\\
&\to
N\int_{-\frac{W}{2}}^{\frac{W}{2}}G(f)\,df.
\end{align}
where (a) is because  $\lambda_d \approx \frac{1}{\delta T}
    G\!\left(\frac{d}{N\delta T}\right)$ for $\delta TW<1$. 
Substituting this into the symbol energy in \eqref{eqn:symbolematch} gives
\begin{equation}
E_s\approx
\frac{P_{{TX}}\delta T}
{\int_{-\frac{W}{2}}^{\frac{W}{2}}G(f)\,df}. \label{eqn:asympE}
\end{equation}
For the $d$th active eigenchannel, the transmit SNR is 
\begin{equation}
\mathrm{SNR}_d
=
\frac{\lambda_dE_s}{\sigma_0^2}
\approx\frac{1}{\delta T}G(f_d)
\frac{P_{{TX}}\delta T}
{\int_{-\frac{W}{2}}^{\frac{W}{2}}G(f)\,df}=
\frac{P_{{TX}}G(f_d)}
{\sigma_0^2\int_{-\frac{W}{2}}^{\frac{W}{2}}G(f)\,df}.\label{eqn:asympsnrtx}
\end{equation}
By changing variable $u=\sqrt{\lambda_d}z$ and plugging \eqref{eqn:asympsnrtx} into \eqref{eqn:itheigenchnlcap}, the asymptotic behavior of the constrained capacity of the $d$th eigenchannel is shown in \eqref{eqn:asympdeigenchnlbehav}.
\begin{figure*}
\begin{equation}
I_d
=
\log_2 M
-\frac{1}{M}\sum_{m=0}^{M-1}
\int_{-\infty}^{+\infty}
\frac{e^{-z^2/(2\sigma_0^2)}}{2\pi\sigma_0^2}
\log_2
\left(
\sum_{\ell=0}^{M-1}
\exp\!\left(
-\frac{1}{
2\sigma_0^2
}
\left|z-\sqrt{
\frac{P_{{TX}}G(f_d)}
{\int_{-\frac{W}{2}}^{\frac{W}{2}}G(v)\,dv}
}\,(s[m]-s[\ell])\right|^2-z^2
\right)
\right)
\,dz. \label{eqn:asympdeigenchnlbehav}
\end{equation}
\end{figure*}
The overall spectral efficiency is the average over the active eigenchannels, which is
\begin{equation}
C'(\delta)=\frac{1}{N_r}\sum_{d=0}^{N_r-1} I_d,
\end{equation}
we therefore obtain the asymptotic constrained capacity of DFT-precoded FTN in \eqref{eqn:asympconscap}. We can see from \eqref{eqn:asympconscap} that asymptotically $C'(\delta)$ is independent of $\delta$ and the capacity is not equal to the Gaussian capacity.


\begin{figure*}
\begin{equation}
\begin{aligned}
\lim_{\delta\to 0} C'(\delta)
&=
\frac{1}{W}
\int_{-\frac{W}{2}}^{\frac{W}{2}}
\Bigg[
\log_2 M
-\frac{1}{M}\sum_{m=0}^{M-1}
\int_{-\infty}^{+\infty}
\frac{e^{-z^2/(2\sigma_0^2)}}{2\pi\sigma_0^2}  \\
&\qquad\qquad \times
\log_2
\left(
\sum_{\ell=0}^{M-1}
\exp\!\left(
-\frac{1}{
2\sigma_0^2
}
\left|z-\sqrt{
\frac{P_{{TX}}G(f)}
{\int_{-\frac{W}{2}}^{\frac{W}{2}}G(v)\,dv}
}\,(s[m]-s[\ell])\right|^2-z^2
\right)
\right)
\,dz
\Bigg]df .
\end{aligned}
\label{eqn:asympconscap}
\end{equation}
\end{figure*}

\section{Proof for Theorem~\ref{thm:asymprxsnrconscap}}\label{app:app2}

In this appendix, we prove the asymptotic constrained capacity stated in  Theorem~\ref{thm:asymprxsnrconscap}.

Under fixed $\mathsf{SNR_{rx}}$, the transmit power is equal to $P_{TX}\delta$. Considering \eqref{eqn:asympE},  Remark~\ref{remark:txrxsnr} gives
\begin{equation}
E_s\approx
\frac{P_{TX}\delta T}
{\int_{-\frac{W}{2}}^{\frac{W}{2}}G(f)\,df},
\end{equation}
which is independent of $\delta$.
 From the parallel-channel model in \eqref{eqn:eigenchnlexpression}, the SNR for the $d$th active eigenchannel $\gamma_d$ is 
\begin{equation}
\gamma_d=\frac{\lambda_d E_s}{\sigma_0^2}. \label{eqn:snrrxasymp}
\end{equation}
As $\delta$ turns to 0, the active eigenvalues $\lambda_d$ satisfy the approximation 
\begin{equation}
\lambda_d \approx \frac{1}{\delta T}G(f_d),
\qquad
f_d=\frac{d}{N\delta T}.
\end{equation}
Therefore, for every active eigenchannel with $\lambda_d>0$,
\begin{equation}
\gamma_d(\delta)\to\infty,\qquad \delta\to 0.
\end{equation}

 We try to find the asymptotic behavior of the constrained capacity for the $d$th eigenchannel. We start with approximating \eqref{eqn:itheigenchnlcap}. 
In order to simplify the notation, we denote 
\begin{equation}
\tau_d[m,\ell]\triangleq s_d[m]-s_d[\ell].
\label{eqn:delta_def}
\end{equation}
Since each symbol $s_d$ is drawn following an identical distribution, we can simply denote $\tau_d[m,\ell]$ as $\tau_{m\ell}$. 
Then \eqref{eqn:itheigenchnlcap} becomes
\begin{align}
    C_d(\delta)&=\log_2M-\frac{1}{M}\sum_{m=0}^{M-1}\int_{-\infty}^{+\infty}\frac{\exp\left({-\frac{u^2}{2\lambda_d\sigma_0^2}}\right)}{2\pi\lambda_d\sigma_0^2} \notag \\
    &\times\log_2\left(\sum_{k=0}^{M-1}\exp\left(-\frac{|u-\lambda_d\sqrt{E_s}\tau_{m\ell}|^2-u^2}{2\lambda_d\sigma_0^2}\right)\right)du. \label{eqn:substitute}
\end{align}
We now expand the quadratic term inside the exponential:
\begin{align}
-\frac{|u-\lambda_d\sqrt{E_s}\tau_{m\ell}|^2-u^2}{2\lambda_d\sigma_0^2}
&=
-\frac{-2\lambda_d\sqrt{E_s}\tau_{m\ell}u+\lambda_d^2E_s|\tau_{m\ell}|^2}{2\lambda_d\sigma_0^2}
\notag\\
&=
\frac{\tau_{m\ell}\sqrt{E_s}u}{\sigma_0^2}
-\frac{\lambda_dE_s|\tau_{m\ell}|^2}{2\sigma_0^2}.
\label{eqn:exponent_simplified}
\end{align}
Substituting \eqref{eqn:exponent_simplified} into \eqref{eqn:substitute} yields
\begin{align}
C_d(\delta)
&=
\log_2 M
-\frac{1}{M}\sum_{m=0}^{M-1}
\int_{-\infty}^{\infty}
\frac{\exp\left(-\frac{u^2}{2\lambda_d\sigma_0^2}\right)}{2\pi \lambda_d\sigma_0^2} \notag \\
&\times\log_2\!\left(
\sum_{\ell=0}^{M-1}
\exp\!\left(
\frac{\tau_{m\ell}\sqrt{E_s}u}{\sigma_0^2}
-\frac{\lambda_dE_s|\tau_{m\ell}|^2}{2\sigma_0^2}
\right)
\right)\,du.
\label{eq:mi_after_expand}
\end{align}
Next, we separate the desired-symbol term from other terms. When $\ell=m$, we have
\begin{equation}
\tau_{mm}=s_d[m]-s_d[m]=0 .
\label{eq:delta_kk}
\end{equation}
Therefore,
\begin{equation}
\exp\!\left(
\frac{\tau_{mm}\sqrt{E_s}u}{\sigma_0^2}
-\frac{\lambda_dE_s|\tau_{mm}|^2}{2\sigma_0^2}
\right)
=
\exp(0)=1 .
\label{eqn:l_equal_k_term}
\end{equation}
Using \eqref{eqn:l_equal_k_term},  \eqref{eq:mi_after_expand} can be rewritten as
\begin{align}
C_d(\delta)
&=
\log_2 M
-\frac{1}{M}\sum_{m=0}^{M-1}
\int_{-\infty}^{\infty}
\frac{\exp\left(-\frac{u^2}{2\lambda_d\sigma_0^2}\right)}{2\pi \lambda_d\sigma_0^2}\times \notag \\
&\log_2\!\left(
1+
\sum_{\ell\neq m}
\exp\!\left(
\frac{\tau_{m\ell}\sqrt{E_s}u}{\sigma_0^2}
-\frac{\lambda_dE_s|\tau_{m\ell}|^2}{2\sigma_0^2}
\right)
\right)\,du.
\label{eq:mi_correct_wrong}
\end{align}
We normalize the Gaussian integration variable by setting
\begin{equation}
u=\sqrt{\lambda_d}\sigma_0 z,
\qquad
du=\sqrt{\lambda_d}\sigma_0\,dz .
\label{eqn:change_variable}
\end{equation}
Thus, \eqref{eq:mi_correct_wrong} can be rewritten as
\begin{align}
&C_d(\delta)
=
\log_2 M
-\frac{1}{M}\sum_{m=0}^{M-1}
\int_{-\infty}^{\infty}
\frac{\exp\left(-z^2/2\right)}{2\pi} \notag \\
&\times\log_2\!\left(
1+
\sum_{\ell\neq k}
\exp\!\left(
\frac{\sqrt{\lambda_dE}\,\tau_{m\ell}}{\sigma_0}z
-\frac{\lambda_dE_s|\tau_{m\ell}|^2}{2\sigma_0^2}
\right)
\right)\,dz .
\label{eqn:mi_z_form}
\end{align}
Using \eqref{eqn:snrrxasymp}, $C_d(\delta)$ can be written as
\begin{align}
C_d(\delta)
&=
\log_2 M
-\frac{1}{M}\sum_{m=0}^{M-1}\int_{-\infty}^{\infty}\phi(z)\times \notag \\
&\qquad\log_2\!\left(
1+\sum_{\ell\neq m}\exp\!\big(\Psi_{m\ell}(z,\gamma_d)\big)
\right)\,dz,\label{eqn:93}
\end{align}
where $\phi(z)=\frac{1}{2\pi}e^{-z^2/2}$ and 
\begin{equation}
\Psi_{m\ell}(z,\gamma_d)
=
\sqrt{\gamma_d}\,\tau_{m\ell} z
-\frac{\gamma_d}{2}|\tau_{m\ell}|^2.
\end{equation}
We also define 
\begin{equation}
L_{m,d}(z)=
\log_2\!\left(
1+\sum_{\ell\neq m}
\exp\!\left(
\sqrt{\gamma_d}\,\tau_{m\ell}z-\frac{\gamma_d}{2}|\tau_{m\ell}|^2
\right)
\right). \label{eqn:defLmdz}
\end{equation}
{To show that \(C_d(\delta)\to \log_2 M\), it remains to prove that the
second term in \eqref{eqn:93} vanishes, i.e.,
\begin{equation}
\int_{-\infty}^{+\infty} f_{\delta}(z)\,dz \to 0,
\qquad
f_{\delta}(z)=\phi(z)L_{m,d}(z).
\end{equation}
The difficulty is that the integration region is unbounded and the
integrand depends on \(\delta\) through \(\gamma_d\). We therefore prove
the convergence by separating the integral into two parts. First, the
tail region \(|z|>A\) is made uniformly small by an integrable upper
bound that does not depend on \(\delta\). Second, on the bounded interval
\(|z|\le A\), the logarithmic term \(L_{m,d}(z)\) is shown to converge
uniformly to zero as \(\gamma_d\to\infty\).

From \eqref{eqn:defLmdz}, for any \(z\), we have
\begin{equation}
0 \le L_{m,d}(z)
\le
\log_2\left(1+(M-1)e^{|z|^2/2}\right).
\end{equation}
Hence,
\begin{equation}
0\le f_{\delta}(z)\le g(z),
\end{equation}
where
\begin{equation}
g(z)=\phi(z)\left(\log_2 M+\frac{|z|^2}{2\ln 2}\right).
\end{equation}
Since \(g(z)\) is integrable, for any \(\varepsilon>0\), there exists
\(A>0\) such that
\begin{equation}
\int_{|z|>A} g(z)\,dz < \varepsilon .
\end{equation}
Therefore, the contribution from the tail region satisfies
\begin{equation}
0\le
\int_{|z|>A} f_{\delta}(z)\,dz
\le
\int_{|z|>A} g(z)\,dz
<\varepsilon . \label{eqn:needtocite1}
\end{equation}

It remains to control the integral over the bounded interval
\(|z|\le A\). For any \(\ell\ne m\) and \(|z|\le A\), we have 
\begin{align}
\sqrt{\gamma_d}\tau_{m\ell}z-\frac{\gamma_d}{2}|\tau_{m\ell}|^2
&\le
\sqrt{\gamma_d}|\tau_{m\ell}|A
-\frac{\gamma_d}{2}|\tau_{m\ell}|^2 .
\end{align}
As \(\delta\to 0\), we have \(\gamma_d\to\infty\). The negative term
of order \(\gamma_d\) dominates the positive term of order
\(\sqrt{\gamma_d}\), and therefore
\begin{equation}
\sqrt{\gamma_d}|\tau_{m\ell}|A
-\frac{\gamma_d}{2}|\tau_{m\ell}|^2
\to -\infty .
\end{equation}
Consequently,
\begin{equation}
\exp\left(
\sqrt{\gamma_d}\tau_{m\ell}z
-\frac{\gamma_d}{2}|\tau_{m\ell}|^2
\right)\to 0
\end{equation}
uniformly for \(|z|\le A\). Since the summation over
\(\ell\ne m\) contains only finitely many constellation points, it follows
that
\begin{equation}
\sup_{|z|\le A} L_{m,d}(z)\to 0 .
\end{equation}
Thus,
\begin{align}
0\le
\int_{|z|\le A} f_{\delta}(z)\,dz
&=
\int_{|z|\le A} \phi(z)L_{m,d}(z)\,dz  \\
&\le
\sup_{|z|\le A}L_{m,d}(z)
\int_{|z|\le A}\phi(z)\,dz  \\
&\le
\sup_{|z|\le A}L_{m,d}(z)
\to 0 .
\end{align}
Hence, for sufficiently small \(\delta\),
\begin{equation}
\int_{|z|\le A} f_{\delta}(z)\,dz < \varepsilon . \label{eqn:needtocite2}
\end{equation}
Combining the bounded interval \eqref{eqn:needtocite1} and the tail region \eqref{eqn:needtocite2} gives
\begin{equation}
0\le
\int_{-\infty}^{+\infty} f_{\delta}(z)\,dz
=
\int_{|z|\le A} f_{\delta}(z)\,dz
+
\int_{|z|>A} f_{\delta}(z)\,dz
<2\varepsilon .
\end{equation}
Since \(\varepsilon>0\) is arbitrary, we obtain
\begin{equation}
\int_{-\infty}^{+\infty} f_{\delta}(z)\,dz\to 0 .
\end{equation}
Therefore,
\begin{equation}
\lim_{\delta\to 0} C_d(\delta)=\log_2 M
\end{equation}
for every active eigenchannel. Finally, using \eqref{eqn:aveovereigenchnl},
 we conclude that
\begin{equation}
\lim_{\delta\to 0} C'(\delta)=\log_2 M .
\end{equation}}

\bibliographystyle{IEEEtran}

\bibliography{main}

@ARTICLE{rusek,
  author={Rusek, Fredrik and Anderson, John B.},
  journal={IEEE Trans. Inf. Theory}, 
  title={{Constrained capacities for faster-than-Nyquist signaling}}, 
  year={2009},
  volume={55},
  number={2},
  pages={764-775},
  doi={10.1109/TIT.2008.2009832}}

@article{gray,
year = {2006},
volume = {2},
journal = {Foundations and Trends® in Communications and Information Theory},
title = {Toeplitz and Circulant Matrices: A Review},
doi = {10.1561/0100000006},
issn = {1567-2190},
number = {3},
pages = {155-239},
author = {Robert M. Gray}
}

@ARTICLE{mazo,
  author={Mazo, J. E.},
  journal={The Bell Syst. Tech. J.}, 
  title={{Faster-than-{N}yquist signaling}}, 
  year={1975},
  volume={54},
  number={8},
  pages={1451-1462},
  doi={10.1002/j.1538-7305.1975.tb02043.x}}

@book{goldsmith,
place={Cambridge}, 
title={{Wireless Communications}},
DOI={10.1017/CBO9780511841224},
publisher={Cambridge University Press}, 
author={Goldsmith, Andrea}, 
year={2005}}

@INPROCEEDINGS{property,
  author={Daniel Kim, Yong Jin},
  booktitle={IEEE Wireless Communications and Networking Conference}, 
  title={Properties of faster-than-{N}yquist channel matrices and folded-spectrum, and their applications}, 
  year={2016},
  volume={},
  number={},
  pages={1-7},
  doi={10.1109/WCNC.2016.7565044}}

@ARTICLE{1_R1,
  author={Ishihara, Takumi and Sugiura, Shinya and Hanzo, Lajos},
  journal={IEEE Access}, 
  title={The Evolution of Faster-Than-{N}yquist Signaling}, 
  year={2021},
  volume={9},
  number={},
  pages={86535-86564},
  doi={10.1109/ACCESS.2021.3088997}}

@ARTICLE{svd,
  author={Ishihara, Takumi and Sugiura, Shinya},
  journal={IEEE Trans. Wirel. Commun.}, 
  title={{SVD}-Precoded Faster-Than-{N}yquist Signaling With Optimal and Truncated Power Allocation}, 
  year={2019},
  volume={18},
  number={12},
  pages={5909-5923},
  doi={10.1109/TWC.2019.2940632}}

@article{zhang2022faster,
  title={Faster-than-{N}yquist signaling for {MIMO} communications},
  author={Zhang, Zichao and Yuksel, Melda and Yanikomeroglu, Halim},
  journal={IEEE Trans. Wirel. Commun.},
  volume={22},
  number={4},
  pages={2379--2392},
  year={2022},
  publisher={IEEE}
}

@ARTICLE{asympftn,
  author={Yoo, Young Geon and Cho, Joon Ho},
  journal={IEEE Commun. Lett.}, 
  title={ Asymptotic Optimality of Binary Faster-than-{N}yquist Signaling}, 
  year={2010},
  volume={14},
  number={9},
  pages={788-790},
  doi={10.1109/LCOMM.2010.072910.100499}}

@ARTICLE{ungerboeckconstrained,
  author={Ungerboeck, G.},
  journal={IEEE Transactions on Information Theory}, 
  title={Channel coding with multilevel/phase signals}, 
  year={1982},
  volume={28},
  number={1},
  pages={55-67},
  doi={10.1109/TIT.1982.1056454}}

@Article{whatshould6Gbe,
author={Dang, Shuping
and Amin, Osama
and Shihada, Basem
and Alouini, Mohamed-Slim},
title={What should 6{G} be?},
journal={Nature Electronics},
year={2020},
month={Jan},
day={01},
volume={3},
number={1},
pages={20-29},
issn={2520-1131},
doi={10.1038/s41928-019-0355-6},
url={https://doi.org/10.1038/s41928-019-0355-6}
}

@INPROCEEDINGS{timelocalization,
  author={Gattami, Ather and Ringh, Emil and Karlsson, Johan},
  booktitle={2015 IEEE Global Communications Conference (GLOBECOM)}, 
  title={Time Localization and Capacity of Faster-Than-{N}yquist Signaling}, 
  year={2015},
  volume={},
  number={},
  pages={1-7},
  doi={10.1109/GLOCOM.2015.7417358}}

@INPROCEEDINGS{precodeftnfschnl,
  author={Ishihara, Takumi and Sugiura, Shinya},
  booktitle={2021 IEEE International Conference on Communications Workshops (ICC Workshops)}, 
  title={Precoded Faster-than-{N}yquist Signaling with Optimal Power Allocation in Frequency-Selective Channel}, 
  year={2021},
  volume={},
  number={},
  pages={1-6},
  doi={10.1109/ICCWorkshops50388.2021.9473860}}

@INPROCEEDINGS{eigenvaluedecompftn,
  author={Chaki, Prakash and Ishihara, Takumi and Sugiura, Shinya},
  booktitle={2021 IEEE International Symposium on Information Theory (ISIT)}, 
  title={Eigenvalue Decomposition Precoded Faster-Than-{N}yquist Transmission of Index Modulated Symbols}, 
  year={2021},
  volume={},
  number={},
  pages={3279-3284},
  doi={10.1109/ISIT45174.2021.9518263}}

@ARTICLE{asynmacftn,
  author={Zhang, Zichao and Yuksel, Melda and Guvensen, Gokhan M. and Yanikomeroglu, Halim},
  journal={IEEE Communications Letters}, 
  title={Capacity Region of Asynchronous Multiple Access Channels With {FTN}}, 
  year={2023},
  volume={27},
  number={7},
  pages={1719-1723},
  doi={10.1109/LCOMM.2023.3269739}}

@ARTICLE{mismatchorig,
  author={Lapidoth, A.},
  journal={IEEE Transactions on Information Theory}, 
  title={Mismatched decoding and the multiple-access channel}, 
  year={1996},
  volume={42},
  number={5},
  pages={1439-1452},
  doi={10.1109/18.532884}}

@ARTICLE{Gokhanmismatch,
  author={{\"U}{\c{c}}{\"u}nc{\"u}, Ali Bulut and
        G{\"u}vensen, G{\"o}khan M. and
        Y{\i}lmaz, Ali {\"O}zg{\"u}r},
  journal={IEEE Transactions on Communications}, 
  title={A Reduced Complexity {U}ngerboeck Receiver for Quantized Wideband Massive {SC-MIMO}}, 
  year={2021},
  volume={69},
  number={7},
  pages={4921-4936},
  doi={10.1109/TCOMM.2021.3071537}}

@ARTICLE{airmis,
  author={Merhav, N. and Kaplan, G. and Lapidoth, A. and Shamai Shitz, S.},
  journal={IEEE Transactions on Information Theory}, 
  title={On information rates for mismatched decoders}, 
  year={1994},
  volume={40},
  number={6},
  pages={1953-1967},
  doi={10.1109/18.340469}}

@ARTICLE{capmis,
  author={Csisz{\'a}r, I. and Narayan, P.},
  journal={IEEE Transactions on Information Theory}, 
  title={Channel capacity for a given decoding metric}, 
  year={1995},
  volume={41},
  number={1},
  pages={35-43},
  doi={10.1109/18.370120}}

@INPROCEEDINGS{rusekair,
  author={Rusek, Fredrik and Fertonani, Dario},
  booktitle={2009 IEEE International Symposium on Information Theory}, 
  title={Lower bounds on the information rate of intersymbol interference channels based on the {U}ngerboeck observation model}, 
  year={2009},
  volume={},
  number={},
  pages={1649-1653},
  doi={10.1109/ISIT.2009.5205784}}

@INPROCEEDINGS{airbicm,
  author={Martinez, Alfonso and Guillen i. Fabregas, Albert and Caire, Giuseppe and Willems, Frans},
  booktitle={2008 IEEE International Symposium on Information Theory}, 
  title={Bit-interleaved coded modulation revisited: A mismatched decoding perspective}, 
  year={2008},
  volume={},
  number={},
  pages={2337-2341},
  doi={10.1109/ISIT.2008.4595408}}

@ARTICLE{ebrahimlowcomplx,
  author={Ibrahim, Ahmed and Bedeer, Ebrahim and Yanikomeroglu, Halim},
  journal={IEEE Open Journal of the Communications Society}, 
  title={A Novel Low Complexity Faster-than-{N}yquist {(FTN)} Signaling Detector for Ultra High-Order {QAM}}, 
  year={2021},
  volume={2},
  number={},
  pages={2566-2580},
  doi={10.1109/OJCOMS.2021.3126805}}

@ARTICLE{Ebrahimlowcomplxclass,
  author={Abbasi, Sina and Bedeer, Ebrahim},
  journal={IEEE Communications Letters}, 
  title={Low Complexity Classification Approach for Faster-Than-{N}yquist {(FTN)} Signaling Detection}, 
  year={2023},
  volume={27},
  number={3},
  pages={876-880},
  doi={10.1109/LCOMM.2023.3236953}}

@article{yu2017low,
  title={Low-complexity graph-based turbo equalisation for single-carrier and multi-carrier {FTN} signalling},
  author={Yu, Tianhang and Zhao, Minjian and Zhong, Jie and Zhang, Jian and Xiao, Pei},
  journal={IET Signal Processing},
  volume={11},
  number={7},
  pages={838--845},
  year={2017},
  publisher={Wiley Online Library}
}

@INPROCEEDINGS{pinarlowcomplx,
  author={{\c{S}}en, P{\i}nar and
        Akta{\c{s}}, Tu{\u{g}}can and
        Y{\i}lmaz, A. {\"O}zg{\"u}r},
  booktitle={2014 IEEE Wireless Communications and Networking Conference (WCNC)}, 
  title={A low-complexity graph-based {LMMSE} receiver designed for colored noise induced by {FTN}-signaling}, 
  year={2014},
  volume={},
  number={},
  pages={642-647},
  doi={10.1109/WCNC.2014.6952123}}

@ARTICLE{linearprecftn,
  author={Li, Qiang and Gong, Feng-Kui and Song, Pei-Yang and Li, Guo and Zhai, Sheng-Hua},
  journal={IEEE Transactions on Broadcasting}, 
  title={Beyond {DVB-S2X}: Faster-Than-{N}yquist Signaling With Linear Precoding}, 
  year={2020},
  volume={66},
  number={3},
  pages={620-629},
  doi={10.1109/TBC.2019.2960941}}

@ARTICLE{precftnotfs,
  author={Hong, Zekun and Sugiura, Shinya and Xu, Chao and Hanzo, Lajos},
  journal={IEEE Wireless Communications Letters}, 
  title={Precoded Faster-Than-{N}yquist Signaling Using Optimal Power Allocation for {OTFS}}, 
  year={2025},
  volume={14},
  number={1},
  pages={173-177},
  doi={10.1109/LWC.2024.3491777}}

@ARTICLE{precequalmultipathfad,
  author={Wen, Shan and Liu, Guanghui and Liu, Chengxiang and Qu, Huiyang and Zhang, Lei and Imran, Muhammad Ali},
  journal={IEEE Transactions on Vehicular Technology}, 
  title={Joint Precoding and Pre-Equalization for Faster-Than-{N}yquist Transmission Over Multipath Fading Channels}, 
  year={2022},
  volume={71},
  number={4},
  pages={3948-3963},
  doi={10.1109/TVT.2022.3146423}}

@ARTICLE{kim2026fasterthannyquistsignalingfinitetimebandwidth,
  author={Kim, Yong Jin Daniel},
  journal={IEEE Transactions on Communications}, 
  title={Faster-Than-{N}yquist Signaling in the Finite Time-Bandwidth Product Regime}, 
  year={2026},
  volume={74},
  number={},
  pages={5605-5618},
  doi={10.1109/TCOMM.2026.3668158}}

@ARTICLE{modelintasgausnoise,
  author={Medard, M.},
  journal={IEEE Transactions on Information Theory}, 
  title={The effect upon channel capacity in wireless communications of perfect and imperfect knowledge of the channel}, 
  year={2000},
  volume={46},
  number={3},
  pages={933-946},
  doi={10.1109/18.841172}}

@inproceedings{ofdmftn,
  title={Bit and Power Allocation Scheme for Multicarrier Faster-Than-{N}yquist Signaling},
  author={Hao Duan and Yuanyuan Gao and Mingxi Guo and Yuehong Shen},
  year={2016/01},
  booktitle={Proceedings of the 2016 International Conference on Intelligent Control and Computer Application},
  pages={367-370},
  issn={2352-538X},
  isbn={978-94-6252-154-4},
  doi={10.2991/icca-16.2016.87},
  publisher={Atlantis Press}
}

@INPROCEEDINGS{ftnadaptivetimedomain,
  author={Liu, Mengmeng and Li, Shuangyang and Li, Qian and Bai, Baoming},
  booktitle={2018 10th International Conference on Wireless Communications and Signal Processing (WCSP)}, 
  title={Faster-than-{N}yquist Signaling Based Adaptive Modulation and Coding}, 
  year={2018},
  volume={},
  number={},
  pages={1-5},
  doi={10.1109/WCSP.2018.8555911}}

@article{zhang2026pushinglimitsunlockingpotential,
      title={Pushing the Limits: Unlocking the Potential of Faster-than-{N}yquist Signaling}, 
      author={Zichao Zhang and Melda Yuksel and Shuangyang Li and Gokhan M. Guvensen and Halim Yanikomeroglu},
      journal={arXiv},
      year={2026},
      eprint={2606.19226},
      archivePrefix={arXiv},
      primaryClass={eess.SP},
}

@ARTICLE{zhang2025iapr,
  author={Zhang, Zichao and Yuksel, Melda and Guvensen, Gokhan M. and Yanikomeroglu, Halim},
  journal={IEEE Transactions on Wireless Communications}, 
  title={Capacity and {IAPR} Analysis for {MIMO} Faster-Than-{N}yquist Signaling With High Acceleration Rate}, 
  year={2026},
  volume={25},
  number={},
  pages={1451-1466},
  doi={10.1109/TWC.2025.3591044}}

\end{document}